\documentclass[11pt]{article}

\usepackage[a4paper,left=15mm,top=20mm,right=15mm,bottom=20mm]{geometry}
\usepackage{enumitem}

\usepackage{amsmath, amsfonts, amssymb, amsthm}
\usepackage[linesnumbered,boxed,noend]{algorithm2e}
\usepackage{authblk}

\usepackage{graphicx}  
\usepackage{mathtools}
\usepackage{subfig}
\usepackage[mathscr]{euscript}

\usepackage{color}
\definecolor{gray}{RGB}{180,180,180}

\usepackage[font=footnotesize,width=.85\textwidth,labelfont=bf]{caption}
\DeclareMathOperator{\lca}{lca}

\renewcommand{\b}{\beta}
\newcommand{\s}{\sigma}
\renewcommand{\S}{\mathcal{S}}

\newcommand{\C}{\mathcal{C}}

\newcommand{\tg}[1]{\tau(#1)}
\newcommand{\ov}[1]{\overline{#1}}
\newcommand{\trans}[1]{\blacktriangle(#1)}
\renewcommand{\sup}[1]{\mathcal{T}_{#1}}

\providecommand{\keywords}[1]{\textbf{\textit{Keywords: }} #1}

\newtheorem{theorem}{Theorem}
\newtheorem{lemma}{Lemma}
\newtheorem{proposition}{Proposition}
\newtheorem{corollary}{Corollary}

\newtheorem{definition}{Definition}

\begin{document}
\setlength{\marginparsep}{-0.6cm}
\setlength{\marginparwidth}{2.3cm}

\title{Predicting Horizontal Gene Transfers with Perfect Transfer Networks}

\author[1]{Alitzel L{\'o}pez S{\'a}nchez}
\author[1]{Manuel Lafond}

\affil[1]{\footnotesize Department of Computer Science, University of Sherbrooke, J1K2R1 Quebec, Canada.}

\date{}
\normalsize

\maketitle

\abstract{   
\textbf{Background:} 
Horizontal gene transfer inference approaches are usually based on gene sequences: parametric methods search for patterns that deviate from a particular genomic signature, while phylogenetic methods use sequences to reconstruct the gene and species trees.
However, it is well-known that sequences have difficulty identifying ancient transfers since mutations have enough time to erase all evidence of such events.  
In this work, we ask whether character-based methods can predict gene transfers.  Their advantage over sequences is that homologous genes can have low DNA similarity, but still have retained enough important common motifs that allow them to have common character traits, for instance the same functional or expression profile.  A phylogeny that has two separate clades that acquired the same character independently might indicate the presence of a transfer even in the absence of sequence similarity.

\textbf{Our contributions:} 
We introduce perfect transfer networks, which are phylogenetic networks that can explain the character diversity of a set of taxa under the assumption that characters have unique births, and that once a character is gained it is rarely lost. Examples of such traits include transposable elements, biochemical markers and emergence of organelles, just to name a few. We study the differences between our model and two similar models: perfect phylogenetic networks and ancestral recombination networks.
Our goals are to initiate a study on the structural and algorithmic properties of perfect transfer networks.  We then show that in polynomial time, one can decide whether a given network is a valid explanation for a set of taxa, and show how, for a given tree, one can add transfer edges to it so that it explains a set of taxa.  We finally provide lower and upper bounds on the number of transfers required to explain a set of taxa, in the worst case.
}

\bigskip
\noindent

\sloppy

\section*{Introduction}

Evolution has historically been seen as a tree-like process in which genetic material is inherited through vertical descent.  However, it is now established that co-existing species from most kingdoms of life, if not all, have exchanged genetic material laterally through hybridation or horizontal gene transfer (HGT).  The latter is well-known to occur routinely between procaryotes~\cite{koonin2001horizontal,thomas2005mechanisms}, but is believed to have affected eucaryotes as well~\cite{keeling2008horizontal,hotopp2011horizontal}.  
HGT is also known to occur between viruses and their hosts~\cite{irwin2022systematic}, between mitochondria and the nucleus~\cite{anselmetti2021gene}, and between tumor cells~\cite{trejo2012cancer}.

Since HGTs play a significant role in shaping evolution, several bioinformatics approaches have been developed to identify them.  Most of these can be classified as either parametric or phylogenetic.
Parametric methods are based on the sequence of one genome of interest and attempt to find DNA regions that exhibit a signature that is different from the rest of the genome (see~\cite{ravenhall2015inferring}).
Phylogeny-based methods consist of taking a set of taxa (e.g. genes and/or species), reconstructing their phylogenetic tree, and inferring the unseen transfer locations on the tree.  A common way of achieving this is through \emph{reconciliation}, which aims to explain the discrepancies between a gene tree and a species tree by finding where the gene duplication, loss, or transfer events occurred~\cite{bansal2012efficient,doyon2010efficient,hellmuth2019reconciling,delabre2020evolution}.  
Finding a most parsimonious reconciliation under this model is NP-hard~\cite{tofigh2010simultaneous,kordi2015complexity,jacox2017resolution}, mainly because the inferred transfer locations need to be \emph{time-consistent}, meaning that they must occur between species that may have co-existed. 
In addition, recent fundamental approaches propose to identify pairwise gene relationships to infer transfers.  For instance, irregularities in the pairwise gene distances can pinpoint to possible transfers~\cite{schaller2021indirect}, or predictions of orthlogs, paralogs, and xenologs can help reconstructing a gene tree and a species \emph{network} that explain these relationships~\cite{geiss2018reconstructing,hellmuth2020generalized,lafond2020reconstruction,jones2022consistency}.

The above approaches are all based on gene sequences in one way or another, either to reconstruct the phylogenies or to infer pairwise relationships.
However, it is well-known that sequences have their limits for predicting HGT, especially in the case of ancient transfers~\cite{boto2010horizontal}.
In this work, we ask whether \emph{character-based} approaches can instead be used to predict HGT on a phylogeny.  
A character is a generic term to denote a trait that a taxa may possess or not, which can be morphological or molecular.  A common example of character-based data is gene expression, where a trait corresponds to whether a species expresses a gene or not in a condition of interest~\cite{de1995phenotypic,rawat2008novel,pontes2013configurable}.
A major advantage of using gene expression profiles, and possibly other character traits, over sequence data comes when highly divergent sequences are involved. In~\cite{Alexander2007}, the authors used expression to recover phylogenetic signals better than using only sequence similarity measures.  This could be because the necessary information to coordinate the folding or function of proteins is encoded in a small number of conserved fragments, in which case the two homologous proteins can share a small percentage of sequence similarity.
This can be leveraged to detect HGTs that are hard to find using sequences, since one could hypothesize that two clades that started expressing the same gene independently could have acquired this behavior by transfer.

The task in this setting is, given a set of characters $\C$ and a set of taxa $\S$ that each possess a subset of $\C$, to explain the diversity of $\S$ in a phylogeny.  
Ideally, $\S$ can be explained by a tree in which taxa that possess a common character form a clade, in which case the tree is 
called a \emph{perfect phylogeny}~\cite{bodlaender1992two,fernandez2001perfect,bafna2003haplotyping,iersel2019third}.
When no such perfect phylogeny exists, transfers may be required to explain the data.
We point out that recently, character-based methods have resurfaced in tumor phylogenetics, where they are used to represent whether a tumor clone has acquired a somatic mutation or not~\cite{della2017character,pradhan2018non,malikic2019phiscs,sashittal2021parsimonious}.

Before gene expression and other character-based data can be used to predict HGT, appropriate models and algorithmic frameworks need to be devised.
To our knowledge, character-based approaches have mostly been used to detect \emph{hybridation} events, where two or more species recombinate to produce an hybrid offspring.  In the most popular models, a set of taxa is explained by an \emph{ancestral recombination graph} (ARG), which is a acyclic directed graph in which nodes with multiple parents represent hybrids, and nodes with a single parent represent vertical descent~\cite{wang2001perfect,gusfield2004optimal,gusfield2014recombinatorics}.
The task of finding recombination events is different from that of finding HGTs.  Recombinations create offsprings whose genetic content is merged from the parents without vertical descent being involved directly.  As a result, there is no donor/recipient relationship.  
In the case of transfers, it is important to distinguish which traits were acquired vertically from the parent, and which traits were given by a donor. 

Another model called perfect phylogenetic networks (PPN) was also introduced in~\cite{nakhleh,nakhlehtesis} to study the evolution of languages, but can also be used for biological characters. To our knowledge, this is the first model that attempts to extend the notion of perfect phylogenies to networks. PPNs belong to the class of \emph{tree-based networks} \cite{francis2015phylogenetic,pons2019tree} which capture the idea of an underlying tree on which a set of transfer highways are ``attached''. The \emph{base tree} indicates where vertical descent occurred and the attached transfer edges clearly show where genetic material could have been exchanged. In this model, the characters can have multiple states and a character is compatible if the network contains a tree in which the character is \emph{convex} (i.e. the subgraph induced by nodes with the same state is connected).

Let us also mention that in~\cite{van2010quantifying}, the authors propose a framework to explain gene evolution using HGT on general networks, in order to minimize the number of genes present in the same ancestral species.

\textbf{Our contributions.}  We introduce \emph{perfect transfer networks} (PTN), which are tree-based networks that can explain how each character was acquired/transferred in a given set of taxa.  Our model is a direct generalization of perfect phylogenies to networks, as we use the same set of evolutionary rules.  That is, we require that in the network, a character acquired by an ancestral species is never lost by its vertical descendants as in the Camin-Sokal  parsimony model \cite{Camin1965}, and that each character has a unique origin.  Additionally, a character can only be transferred horizontally on the edges that are explicitly labeled as transfers. It is worth mentioning that in~\cite{avni2020new}, the authors study an HGT inference framework in which characters that admit a perfect phylogeny are ignored, whereas characters that do not are treated as evidence of transfers.  Our work can be seen as an effort to formalize this idea.

We then study the structural and algorithmic aspects of PTNs.  We first show that PTNs have two equivalent definitions that are both generalizations of perfect phylogenies.  We then distinguish PTNs from recombination networks and from perfect phylogenetic networks by showing that some taxa sets are explained by different networks depending on the model.

As for the algorithmic aspects, we study three different problems. First, we ask whether a given tree-based network can explain the characters of a set of taxa and provide a simple, polynomial-time algorithm for the problem.  Second, we study the tree completion problem where, given a tree, we are asked to add transfers to it so that it explains the input taxa.
We show that any tree can explain any set of taxa, even if the characters at the ancestral nodes of the tree are constrained by the input.  
Third, we study the reconstruction problem, where only the taxa are known and we must reconstruct a tree-based network with a minimum number of transfers that explains them.  The algorithmic classification of this problem remains open, but we provide nearly exponential lower and upper bounds on the number of transfers required in the worst case, with respect to the number of characters. 
We then conclude with a discussion on open problems, including the problem of adding a minimum number of transfers to a tree to make it explain a set of taxa.

\section*{Preliminaries}
In this section, we describe the standard phylogenetic notions used in the paper, and then define our perfect transfer network model.

\subsection*{Phylogenetic Networks and tree-based networks}
For an integer $n$, we use the notation $[n] = \{1, \ldots, n\}$.  All graphs in this work are directed and loopless.
A directed graph $G$ is \emph{connected} if the underlying undirected graph of $G$ is connected.
A \emph{binary phylogenetic network}, or simply a \emph{network} for short, is a directed acyclic graph $G = (V, E)$ such that either $|V| = 1$, or such that $G$ satisfies the following conditions:
\begin{itemize}
    \item there is a set of vertices with in-degree 1 and out-degree 0, called \emph{leaves}.
    \item there is a unique vertex with in-degree 0 and out-degree $2$, called the \emph{root}.
    \item every other vertex has either in-degree 1 and out-degree 2 (\emph{tree nodes}), in-degree 2 and out-degree 1 (\emph{reticulation nodes}), or in-degree $1$ and out-degree $1$ (\emph{subdivision nodes}).
\end{itemize}

We say that $(u, v) \in E$ is a \emph{tree edge} if $v$ is a tree node or a leaf.
Note that the usual definition of a network forbids subdivision nodes.  We allow them only because it simplifies some of the definitions and proofs.

For a network $G$, we write $\rho(G)$ for the root of $G$ and $L(G)$ for the leaves. 
If $|V(G)| = 1$, then we define $\rho(G)$ as the single vertex of $G$ and consider that $L(G) = \{\rho(G)\}$.
If $\s$ is a bijection from $L(G)$ to a set $\S$, we call $\s$ an \emph{$\S$-map for $G$}, or just an \emph{$\S$-map} if $G$ is understood. 
Now suppose that $G$ is a directed graph, network or not.
We say that $u \in V(G)$ \emph{reaches} a node $v \in V(G)$ if there exists a directed path from $u$ to $v$ in $G$.
We denote by $R_u(G)$ the set of nodes that $u$ reaches in $G$, and we note that $u \in R_u(G)$.
For a subset $W$ of $V(G)$, we denote by $G[W]$ the subgraph of $G$ induced by $W$.
We will also denote by $G-W$ the graph obtained by the removal of $W$ from $V(G)$ and all of its incident edges.  In other words, $G - W = G[V(G) \setminus W]$.  

A \emph{tree} $T$ is a network whose underlying undirected graph has no cycles.
We say that $W \subseteq V(T)$ \emph{forms a subtree of $T$} if $T[W]$ is a tree.
We say that a vertex $v \in V(T)$ is an ancestor of $u \in V(T)$ if $v$ is on the path from $\rho(T)$ to $u$. In this case, we will call $u$ a descendant of $v$.  Note that $v$ is an ancestor and descendant of itself.
The \emph{ancestor order} $\preceq_T$ is a partial order in a tree $T$. When $u \prec_T v$ we say that $v$ is an \emph{ancestor} of $v$ and $u$ is considered a \emph{descendant} of $v$. In this partial order we have that the root $\rho$ of $T$ is the unique maximal element. We say that two nodes $u$ and $v$ are \emph{comparable} if $u \preceq_T v$ or $v \preceq_T u$. We say that they are \emph{incomparable} otherwise. We will drop the subscript $T$ when $T$ is clear from the context.
For $v \in V(T)$, we will use $T(v)$ to refer to the subtree of $T$ rooted at $v$ (that is, $T(v)$ contains $v$ and all of its descendants). 

A network $G = (V, E)$ is a \textbf{tree-based network} \cite{Pons2018} if $G$ has no subdivision nodes, and there is a partition $\{E_S, E_T\}$ of $E$ such that the subgraph $\sup{G} := (V, E_S)$ is a tree with the same set of leaves as $G$, which is called the \emph{support tree} of $G$.  The edges in $E_S$ are called \emph{support edges} and the edges in $E_T$ are called \emph{transfer edges}.  
Note that $\sup{G}$ contains subdivision nodes, unless $E_T$ is empty.  The tree obtained from $\sup{G}$ by suppressing its subdivision nodes is called the \emph{base tree} of $G$ (suppressing a subdivision node $u$ with parent $p$ and child $v$ consists of removing $u$ and adding an edge from $p$ to $v$).
Roughly speaking, a tree-based network $G$ can be obtained by starting with a tree and inserting transfer edges into it. Note that in most cases, the partition of the edges into $E_s$ and $E_T$ will be known (whereas tree-based networks merely require these to exist). When these edge sets are given, the network is sometimes called an \emph{LGT network}, see~\cite{Cardona2015reconstruction}.

As mentioned in the introduction, networks should be \textit{biologically-feasible} in terms of time. We define a \emph{time consistent map} over a tree-based network $G$ with support edges $E_S$ and transfer edges $E_T$ as a function $\tau: V \to \mathbb{R}$ such that: 
\begin{itemize}
    \item  for every $(u,v) \in E_S$, $\tau(u) > \tau(v)$.
    \item for every $(u,v) \in E_T$, $\tau(u) = \tau(v)$.
\end{itemize}
We say that $G$ is a \emph{time-consistent tree-based network} if there exists a time consistent map for $G$~\cite{FRANCIS201893}. 
Note that the existence of a time-consistent map on a network implies that it is tree-based~\cite{murakami2021phylogenetic} (but the converse does not necessarily hold).
In the following sections, we will assume that all the tree-based networks are time-consistent without explicit mention.

\subsection*{Perfect transfer networks}
We now propose to extend the \emph{perfect phylogeny} model to tree-based networks. Let ${\S = \{S_1, S_2, \dots, S_n\}}$ be a set of taxa and $\C = \{ c_1, c_ 2, \dots, c_m\}$ a set of characters. 
We view a taxa $S_i$ as the set of characters that it possesses, so that for each $i \in [n]$, $S_i$ is a subset of $\C$.  
Our goal is to explain the character diversity of $\S$ using its evolutionary history.  Given a tree-based network $G$ with $\S$-map $\s$, we want to know where each character appeared in $G$ under the conditions that each character has a single origin, that it cannot be lost once acquired, and that it can be transferred. 
Throughout the phylogenetic literature, requiring a single origin is called the \emph{homoplasy-free} assumption (or sometimes the ``no parallel evolution'' or ``no convergent evolution''), which states that characters cannot arise independently in unrelated lineages~\cite{sanderson1996homoplasy,semple2002tree}.  HGT is not considered to be a cause of homoplasy, but of course homoplasy can occur even in the presence of HGT.  Nonetheless, this assumption has historically been used as a first step towards more complex models (see  e.g.~\cite{ringe2002indo}).

To formalize this, given a tree-based network $G$, a \emph{$\C$-labeling} of $G$ is a function $l : V(G) \rightarrow 2^{\C}$ that maps each node of $G$ to the subset of characters that it possesses (here, $2^{\C}$ represents the powerset of $\C$).  
For a character $c \in \C$, we will denote by $V_c(l) = \{v \in V(G) : c \in l(v)\}$ the set of nodes that possess character $c$, and we denote by $\overline{V}_c(l) = V(G) \setminus V_c(l)$ the nodes that do not have it.  If $l$ is clear from the context, then we may simply write $V_c$ and $\ov{V}_c$.

Our evolutionary requirements are encapsulated in the following definition.

\begin{definition}[Perfect transfer networks]\label{def:hgt-explain}
Let $\S$ be a set of taxa on characters $\C$, let $G = (V, E_S \cup E_T)$ be a tree-based network, and let $\s$ be an $\S$-map for $G$.
We say that a $\C$-labeling $l$ of $G$ \emph{explains} $\S$ if the following conditions hold:
\begin{itemize}
    \item 
    for each $v \in L(G)$, $l(v) = \s(v)$;
    
    \item 
    for each support edge $(u, v) \in E_S$, $c \in l(u)$ implies that $c \in l(v)$ (never lost once acquired);
    
    \item 
    for each $c \in \C$, there exists a unique node $v \in V_c(l)$ that reaches every node of $V_c(l)$ in $G[V_c(l)]$ (single origin).

\end{itemize}

Furthermore, we call the pair $(G, \s)$ a \emph{perfect transfer network} (PTN) for $\S$ if there exists a $\C$-labeling of $G$ that explains $\S$.
\end{definition}

See Figure~\ref{fig:main_fig}.2 for an example of a PTN.  Later on in Theorem~\ref{thm:connect}, we will show that Definition~\ref{def:hgt-explain} is similar, though slightly different, to a connectedness requirement known as \emph{convexity} on each character, see~\cite{gusfield2014recombinatorics}.
Notice that if $G$ is a tree, then every edge is a support edge and Definition~\ref{def:hgt-explain} coincides with the definition of a perfect phylogeny from~\cite{gusfield2014recombinatorics}.
If $G = (V, E_S \cup E_T)$ is a tree-based network, the definition does not explicitly state what can or cannot be done with transfer edges.  The way to see this is that the definition \emph{does not forbid} ancestral taxa from using transfer edges.  That is, if $(u, v) \in E_T$, then $u$ can transmit any subset of its characters to $v$ horizontally. The motivation behind the requirements of Definition \ref{def:hgt-explain} is to model the presence and absence of traits that have unique origins and cannot be lost throughout evolutionary processes. Examples of such traits include transposable elements (TEs) which are unique genomic sequences that have integrated into the genome and are rarely lost~\cite{Bourque2018}, biochemical markers such as metabolites which are small molecules that function as intermediates and products of metabolic processes \cite{AltafUlAmin2019}, and emergence of organelles such as mitochondria or chloroplasts that results from endosymbiotic events and is irreversible \cite{Zachar2020,anselmetti2021gene}. It is worth highlighting that the horizontal transfer of TEs between species is a prevalent phenomenon that significantly contributes to their sustained viability over time \cite{Wells2020}. As for metabolites, it has been previously shown that HGT plays a role in the generation of new metabolic pathways in bacteria \cite{Goyal2022}.

We are interested in the following algorithmic problems:
\begin{itemize}
    \item 
    \emph{The PTN-recognition problem:} given a tree-based network $G$ with $\S$-map $\s$, is $(G, \s)$ a PTN for $\S$?  That is, does there exist a $\C$-labeling of $G$ that explains $\S$? See Figure~\ref{fig:main_fig}.4 for an example of a network that is not a PTN.
    
    \item 
    \emph{The tree-completion problem:} given a tree $T$ with $\S$-map $\s$, does there exist a PTN $(G, \s)$ for $\S$ such that $T$ is the base tree of $G$? 
    See Figure~\ref{fig:main_fig}.2.

    We are also interested in the minimization variant of this problem, where we require that $(G, \s)$ has a minimum number of transfer edges.
    See Figure~\ref{fig:main_fig}.3.
    

    \item 
    \emph{The PTN-reconstruction problem:} given a set of characters $\C$ and a set of taxa $\S$, find a PTN $(G, \s)$ for $\S$ with a minimum number of transfer edges.
    
\end{itemize}

\begin{figure}[h!]
    \centering
    \includegraphics[width=0.85\textwidth]{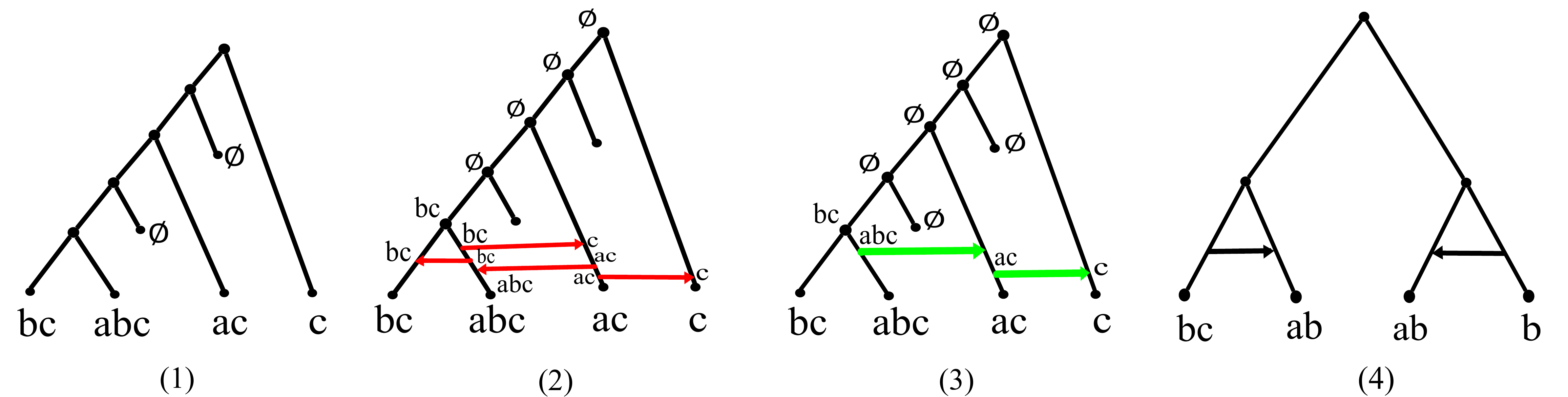}
    \caption{
    (1) A species tree $T$ with a set of taxa $\S$ on characters $\C = \{a,b,c\}$. (2) A PTN with $T$ as base tree that explains $\S$. Red arrows represent transfer edges. (3) Another PTN with $T$ as base tree that also explains $\S$.  Note that this PTN requires two less transfers.
    (4) A tree-based network $G$ with $\S-$map $\s$ for which no labeling can explain $\S$. }
    \label{fig:main_fig}
\end{figure}

We show that the recognition problem can be solved in time $O(|\C||V(G)|^2)$.  For the tree-completion problem, we provide a more general result: any tree with any given pre-labeling can explain any set of taxa.  To be more specific, for any given tree $T$ and any $\C$-labeling of $T$ that satisfies the \emph{never lost once acquired} condition, one can always explain $\S$ by adding transfers in a time-consistent manner while preserving the given labeling.  This motivates the need for the minimization variant, which leads to several open problems.
For the tree reconstruction problem, we give exponential lower and upper bounds on the number of transfers required by a set of $k$ characters, in the worst case.

\section*{Properties of the perfect transfer model}
Before delving into the algorithms, we study our model a bit more in-depth.  First, we provide an alternate definition of perfect transfer networks in terms of character connectedness.  This definition is sometimes easier to deal with in our proofs, and is akin to perfect phylogenies that also admit a similar equivalent definition.  
Second, we ensure that our model does not reinvent the wheel by explicitly stating its differences with other models.

\begin{theorem}\label{thm:connect}
Let $G = (V, E_S \cup E_T)$ be a tree-based network with $\S$-map $\s$.  
Then a $\C$-labeling $l$ of $G$ explains $\S$ if and only if the following conditions hold:
\begin{itemize}
    \item 
    for each $v \in L(G)$, $l(v) = \s(v)$; 
    
    \item 
    for each $c \in \C$, $G[V_c(l)]$ is connected and contains a unique node of in-degree $0$;
    
    \item 
    for each $c \in \C$, either $V = V_c(l)$, or $\sup{G}[\ov{V}_c(l)]$ is connected and contains $\rho(G)$.
    
\end{itemize}
\end{theorem}

\begin{proof}
($\Rightarrow$)
Suppose that $l$ is a $\C$-labeling of $G$ that explains $\S$, according to Definition~\ref{def:hgt-explain}.
We argue that the three conditions stated in the theorem are true.
By Definition~\ref{def:hgt-explain}, $l(v) = \s(v)$ for each $v \in L(G)$ holds.
For the other conditions, let $c \in \C$.
Let $v$ be the unique node of $V_c(l)$ that reaches every node in $G[V_c(l)]$, which is guaranteed to exist by Definition~\ref{def:hgt-explain}.
Because $G[V_c(l)]$ is acyclic, no node other than $v$ can reach $v$ in $G[V_c(l)]$.  This implies that $v$ has in-degree $0$ in $G[V_c(l)]$.  
Moreover, $G[V_c(l)]$ is connected since $v$ reaches all of its nodes.
If $\ov{V}_c(l)$ is empty, then the third condition also holds, so assume this is not the case.  Let us now focus on $\sup{G}[\ov{V}_c(l)]$.
Observe that because characters are never lost once acquired, $l$ satisfies the property that for every $(u, w) \in E_S$, $w \in \ov{V}_c(l)$ implies that $u \in \ov{V}_c(l)$.
This in turn implies that for any $u \in \ov{V}_c(l)$, every ancestor of $u$ in $\sup{G}$ is in $\ov{V}_c(l)$, including the root $\rho(G)$.
Since we assume that $\ov{V}_c(l)$ is non-emtpy, it follows that $\rho(G) \in \ov{V}_c(l)$.
It also follows that $\sup{G}[\ov{V}_c(l)]$ is connected because all of its nodes have a path to $\rho(G)$.

($\Leftarrow$)
Suppose that $l$ satisfies all the conditions of the theorem.  We show that  all properties of Definition~\ref{def:hgt-explain} hold.
For each $v \in L(G)$, we know that $l(v) = \s(v)$.
For the other conditions, let $c \in \C$.
First suppose for contradiction that there is a support edge $(u, v) \in E_S$ such that 
$u \in V_c(l)$ but $v \in \ov{V}_c(l)$. 
Thus $\ov{V}_c(l)$ is not empty, in which case $G[\ov{V}_c(l)]$ is connected and contains $\rho(G)$.
The path in the support tree $\sup{G}$ from $\rho(G)$ to $v$ goes through $u$.  
But $\rho(G) \in \ov{V}_c(l)$, which is a contradiction since $\sup{G}[\ov{V}_c(l)]$ is connected, but here $u$ disconnects $\rho(G)$ from $v$ in $\sup{G}$.  Thus the condition of never losing acquired characters holds.
Now let $v$ be the unique node of $V_c(l)$ of in-degree $0$ in $G[V_c(l)]$.
Assume that there is some $u \in V_c(l)$ that $v$ does not reach in $G[V_c(l)]$.
Let $P_u$ be the set of nodes of $V_c(l)$ that reach $u$ in $G[V_c(l)]$.
Because $G[V_c(l)]$ is acyclic, $P_u$ must contain a node $w$ of in-degree $0$, contradicting that 
$v$ is the unique node of in-degree $0$.
Thus $v$ reaches every node in $G[V_c(l)]$ and, because it has in-degree $0$, it is the unique such node.  
Thus the single origin condition is satisfied.
\end{proof}

\subsection*{Perfect transfer networks versus perfect phylogenetic networks}

Before we move on, it is important to put our model in perspective with Perfect Phylogenetic Networks (PPN), which, to our knowledge, are the closest to our work. The idea of a PPN is that a network contains several evolutionary trees, and a character could evolve in any one of those trees. More specifically, in the PPN model, characters are multi-state, so that a character can be in any of the set of possible states $Z$.  
 A character $c$ is compatible with a tree $T$ whose leaves are labeled by the states of $c$ if there is a state-labeling of the internal nodes of $T$ such that, for each $z \in Z$,  the nodes in state $z$ form a connected subgraph of $T$. Given a network $G$ and a tree $T$, we say that $T$ is \emph{displayed} by $G$ if $T$ can be obtained by successively removing reticulation edges and suppressing subdivision nodes. A tree-based network $G$ is a PPN for a set of characters $\C$ if each character is compatible with some tree displayed by $G$. In terms of character evolution, for PPNs every state is subject to the same evolutionary constraints.  
 In contrast, the character evolution model implied by PTNs represents ``presence'' and ``absence'' of a trait, and these two states have different behavior. 
 This difference plays an important role for transfer edges. The PTN model explicitly prohibits the transfer of an ``absence'' state, whereas for PPNs any state is allowed to be transferred.  See Figure~\ref{fig:ppn_ptn} for an example of a tree-based network that is a PPN, but not a PTN (if the state $0$ is interpreted as absence).
\begin{figure}[h]
    \centering
    \includegraphics[width=0.7\textwidth]{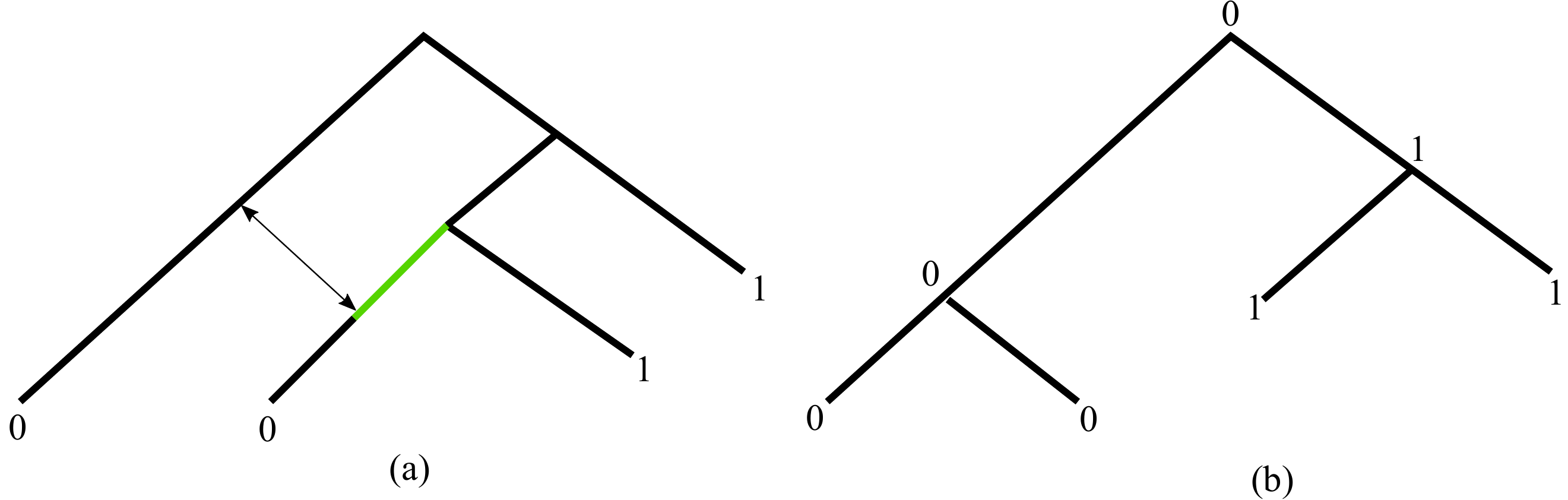}
    \caption{(a) A PPN network with one character $c$ with two states $\{0,1\}$. (b) A tree displayed by $G$ that admits a labeling such that every state forms a connected component. Note that this tree can be obtained by the removal of the green edge and suppression of the resulting subdivision node. The interpretation of this scenario is that state $0$ was transferred using the birectional edge. On the other hand, a solution for PTNs would forbid the transfer of state absence, and would also 
disregard the removal of the green edge, in which case no $\C$-labeling is possible.}
    \label{fig:ppn_ptn}
\end{figure}

Also note that it is NP-hard to decide whether a given network $G$ is a PPN for a given set of characters $\C$~\cite{nakhlehtesis}.  
Later on, we will show that the analogous recognition problem for PTNs can be solved in polynomial time.
However, our result holds for binary character states and the hardness proof for PPNs relies on the character states being non-binary, and it is unclear whether the hardness is preserved for binary character states.

\subsection*{Perfect transfer networks versus recombination networks}
Another model with goals that are similar to ours are recombination networks.
In this model, the indices of the characters $\C = \{c_1, \ldots, c_m\}$ determine an ordering of the characters.  For a string $B$, $B[j]$ denotes its $j$-th character and $B[i .. j]$ its substrings containing positions from $i$ to $j$.  Each taxa $S_i \subseteq \C$ can be represented as an $m$-bit string $\b(S_i)$ in which $\b(S_i)[j] = 1$ 
if and only if $S_i$ possesses character $c_j$.  
Given an $m$-bit string $B$ and an odd integer $d$, we say that $B$ is a \emph{$d$-crossover} of two other $m$-bit strings $X$ and $Y$ if there are indices $i_1, \ldots, i_d$ such 
that $B = X[1 .. i_1] Y[i_1 + 1 .. i_2] X[i_2 + 1 .. i_3] \ldots Y[i_d + 1 .. n]$.  
If $d$ is even, the definition of a $d$-crossover is the same except that the last substring is $X[i_d + 1 .. n]$.   Note that the roles of $X$ and $Y$ are interchangeable.
For a network $G$ and $\S$-map $\s$, a \emph{binary $\C$-labeling} of $G$ is a function $f$ in which 
$f(v)$ is an $m$-bit binary string for each $v \in V(G)$, such that $f(\rho(G))$ only contains $0$s.
A binary $\C$-labeling $f$ of $G$ \emph{explains $\S$ with $d$-crossovers} if $f(v) = \b(\s(v))$ for each $v \in L(G)$, and the following holds:
\begin{itemize}

\item 
for each reticulation node $v$ with parents $u$ and $w$, $f(v)$ is a $d$-crossover of $f(u)$ and $f(w)$;

\item 
for each tree edge $(u, v)$, $f(v)$ is obtained from $f(u)$ by flipping some $0$ positions to $1$ (we may decide to flip none);

\item 
for each $i \in [m]$, there is at most one tree edge $(u, v)$ with $f(u)[i] = 0$ but $f(v)[i] = 1$.
\end{itemize}

We say that $(G, \s)$ is an \emph{ancestral recombination graph} (ARG) for $\S$ with $d$-crossovers if there exists a binary $\C$-labeling of $G$ that explains $\S$ with $d$-crossovers.
We denote by $ARG_d(\S)$ the set of all ARGs for $\S$ with $d$-crossovers.  We also denote $ARG_{\infty}(\S) = \bigcup_{d=1}^{\infty} ARG_d(\S)$.

In the literature, single crossovers are modeled with $d = 1$ and have been studied extensively.  The $d = 2$ case is often referred to as double crossovers and can also model gene conversion.  It was stated in~\cite{gusfield2014recombinatorics} that $d > 6$ is rarely considered in practice. 

The most obvious difference between ARGs and PTNs is that ARGs were introduced to model hybridation events, whereas we created PTNs to model horizontal gene transfer.  
More specifically, ARGs do not differentiate between the two parents of a reticulation node, whereas in our model, one parent only transmits genetic content via vertical descent whereas the other does so via HGT.  In fact, ARGs do not need to be tree-based networks.
Although this is a fundamental difference, it is also interesting to ask whether, among the class of tree-based networks, the data explanation depends on the model.

To formalize this, we write $PTN(\S)$ for the set of perfect transfer networks for $\S$.
The next result shows that ARGs with $d$-crossovers are incomparable with PTNs unless we allow an arbitrary number of crossovers.  We emphasize that even though infinite crossovers can emulate transfers, they still cannot distinguish vertical from horizontal inheritance.

\begin{figure}[h]
    \centering
    \includegraphics[width=0.9\textwidth]{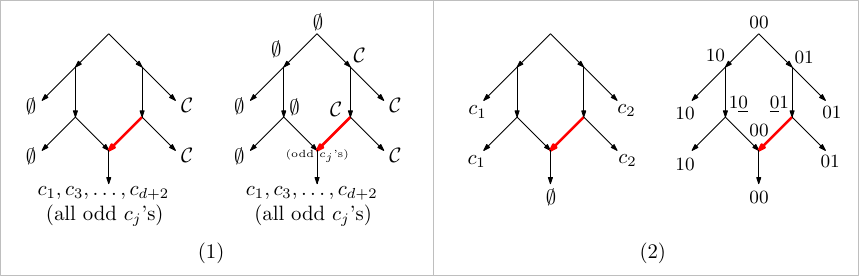}
    \caption{(1) Left: a network $G$ with $5$ taxa at the leaves.  The fat red edge is a transfer edge.  Two taxa have no character, two have all characters of $\C$, and one has only the odd numbered taxa (we assume odd $d$ in the figure).  Right: a $\C$-labeling that explains $G$ (the reticulation receives the odd characters from the transfer edge).  This network has no explanation with $d$-crossovers.  (2) Left: a network $G$ with $5$ taxa on character set $\C = \{c_1, c_2\}$.  Right: a binary $\C$-labeling with single crossovers that explains $G$.  This network is not a PTN.}
    \label{fig:ptn-arg1}
\end{figure}

\begin{proposition}
\label{prop:ptnvsarg}
 The following relationships between PTNs and ARGs hold:
\begin{itemize}
\item 
for any fixed $d \geq 1$, there exists a set of taxa $\S$ on $d + 2$ characters such that $PTN(\S) \setminus ARG_d(\S)$ is non-empty;

\item  
there exists a set of taxa $\S$ on two characters such that $ARG_1(\S) \setminus PTN(\S)$ is non-empty;

\item 
for any set of taxa $\S$, $PTN(\S) \subseteq ARG_{\infty}(\S)$.
\end{itemize}
\end{proposition}

\begin{proof}
For (1), consider the network $G$ in Figure~\ref{fig:ptn-arg1}.1.  For any $\C$, the labeled network shows that $(G, \s)$ is in $PTN(\S)$.  Put $\C = \{c_1, \ldots, c_{d + 2}\}$.  We argue that $(G, \s)$ is not in $ARG_d(\S)$.  
Note that in any binary $\C$-labeling of $G$ that explains $\S$, the parent of each $\emptyset$ leaf taxa must be the string $00 \ldots 0$, as otherwise they would need to transmit a $1$ to these leaves.
Likewise, the parent of the upper $\C$ leaf taxon must be $11 \dots 1$.  If not, there would be a $0-1$ flip on the edge leading to the upper $\C$.  But, we would also need another $0-1$ flip somewhere on the path to the lower $\C$ taxon, which is not allowed in ARGs.  
So, the parents of the single reticulation have labels $00 \ldots 0$ and $11 \ldots 1$.  Moreover, the root is labeled $00 \ldots 0$, and thus all the possible $0-1$ flips occur between the root and its right child.  We cannot have any new $0-1$ flip, and so the reticulation must be labeled $0101 \ldots 01$ to be able to transmit the required characters to the odd-characters leaf (or the last character is $0$ if $d$ is even).  Whether $d$ is odd or even, there are $d + 2$ characters and we must alternate each of them between the parents.  This requires $d + 1$ crossovers, and so $(G, \s) \notin ARG_d(\S)$.

For (2), consider the network $G$ in Figure~\ref{fig:ptn-arg1}.2. 
There are only two characters $c_1, c_2$.
The right side of the figure shows that $G$ can be explained under single crossovers.  However, we can argue that $(G, \s)$ is not a perfect transfer network.  This is because if there is a $\C$-labeling $l$ that explains the taxa at the leaves, the parent of each $c_1$ leaf must contain $c_1$.  But by the \emph{never lose once acquired} condition, $c_1$ should then be transmitted vertically to the leaf that has no character, a contradiction.  Thus $(G, \s) \notin PTN(\S)$.

For (3), let $(G, \s) \in PTN(\S)$.  Let $l$ be a $\C$-labeling of $G$ that explains $\S$.  To show that $(G, \s) \in ARG_{\infty}(\S)$, we need to modify $l$ slightly.  Specifically, we ensure that characters never have their first-appearance on a reticulation node.
Let $v$ be a reticulation node of $G$ with parents $u, w$ such that $(u, v)$ is a support edge and $(w, v)$ is a transfer edge.
Suppose that there is $c_j \in \C$ such that $c_j \notin l(u), c_j \notin l(w)$, but $c_j \in l(v)$.  
It follows from Theorem~\ref{thm:connect} that $v$ is the unique node of in-degree $0$ in $G[V_{c_j}(l)]$.  
Consider the labeling $l'$ obtained by taking $l$, but applying $l'(v) = l(v) \setminus \{c_j\}$.
Let $v'$ be the unique child of $v$.  One can easily see that $G[V_{c_j}(l')]$ is connected and that $v'$ is its unique node of in-degree $0$.  Moreover, $\sup{G}[\ov{V}_{c_j}(l')]$ is the same as $\sup{G}[\ov{V}_{c_j}(l)]$ but with $v$ added.  It is still connected since $v$ is a child of $u \in \ov{V}_{c_j}(l)$, and it still contains the root.  
Hence by Theorem~\ref{thm:connect}, $l'$ also explains $\S$.  

By applying the above argument to every character, we obtain a labeling $l$ that explains $\S$ such that for any character $c_j$, the first-appearance of $c_j$ under $l$ does not occur at a reticulation node.  Assume that $l$ has this property.

Consider the binary $\C$-labeling $f$ of $G$ in which, for each $v \in V(G)$, 
we put $f(v)[j] = 1$ if and only if $c_j \in l(v)$.  
We claim that $f$ explains $\S$ with $m$-crossovers, where $m = |\C|$.  
First consider a reticulation node $v$ with parents $u$ and $w$, where $(u, v)$ is a support edge and $(w, v)$ is a transfer edge.  We must argue that $f(v)$ is an $m$-crossover of $f(u)$ and $f(w)$.
Because under $l$, no character first-appears on a reticulation, we have that for each $c_j \in l(v)$, either $c_j \in l(u)$ or $c_j \in l(w)$.
Moreover, for each $c_j \notin l(v)$, we must have $c_j \notin l(u)$ (by the \emph{never lost once acquired} condition).
In terms of $f$, this means that for each $j \in [m]$,
if $f(v)[i] = 1$ then one of $u$ or $w$ also has a $1$ in position $j$, and if $f(v) = 0$, then $f(u)[j] = 0$. 
It follows that $f(v)$ can be obtained with at most $m$ crossovers from its parents.

Next, we argue that for each character $c_j$, there is at most one tree edge 
$(u, v)$ that flips $c_j$ from $0$ to $1$ under $f$.  If there were two such edges, then under $l$, there would be two trees nodes that possess $c_j$ but not their (unique) parent.  Thus $G[V_{c_j}(l)]$ would have two nodes of in-degree $0$, contradicting Theorem~\ref{thm:connect}.
Third, the condition that characters are not lost after being acquired implies that, for every tree edge $(u, v)$, 
$f(v)$ can be obtained from $f(u)$ by flipping $0$s to $1$s, but not vice-versa (this follows from the fact that tree edges are support edges).
Thus, $f$ explains $\S$ and we conclude that $(G, \s) \in ARG_{\infty}(\S)$.
\end{proof}

\section*{Algorithmic Problems}

We now study the algorithmic aspects of recognition, completion, and reconstruction of PTNs.

\subsection*{Recognizing perfect transfer networks}
The first problem in this section is analogous to the perfect phylogeny problem, where we must find a labeling of all the inner nodes so that the network correctly represents the evolution of a given set of species.  This is not as trivial as in the tree case, since their may be multiple options for the originator of a character $c \in \C$.

\medskip 

\noindent 
The {\sf PTN recognition} problem\\
\textbf{Input.}  A set of taxa $\S$ on characters $\C$ and a tree-based network $G = (V, E_S \cup E_T)$ with $\S$-map $\sigma$. \\
\textbf{Output.} A $\C$-labeling of $G$ that explains $\S$, if one exists.

\vspace{4mm}

Importantly, the above problem formulation lets us assume that we have knowledge of support and transfer edges.  
We first need some intermediate results (recall that $\sup{G}$ is the support tree of $G$).

\begin{lemma}\label{lem:Fzero}
Let $G = (V,E_S \cup E_T)$ be a tree-based network with $\S$-map $\s$, and let $l$ be a $\C$-labeling of $G$ that explains $\S$. 
Let $u \in \ov{V}_c(l)$.  Then for each ancestor $v$ of $u$ in $\sup{G}$, we have $v \in \overline{V_c}(l)$ as well.
\end{lemma}

\begin{proof}
Suppose that $u \in \ov{V}_c(l)$ has an ancestor $v$ in $\sup{G}$ such that $v \in V_c(l)$. Consider the unique path of $\sup{G}$ from $v$ to $u$, namely the sequence $(v,v_1, v_2,\dots , v_k, u)$. Due to the \emph{never lost once acquired} requirement of Definition~\ref{def:hgt-explain} $\{v_1,v_2, \dots , v_k\} \subseteq V_c(l)$. In this way the edge $(v_k,u)$ will represent the loss of character $c$, a contradiction.
\end{proof}

In other words, any node that has at least one descendant in $\overline{V_c}(l)$ should be in $\overline{V_c}(l)$ in every possible $\C-$labeling $l$ that explains $\S$. 
Since we must have $l(v) = \s(v)$ for each leaf $v$, we can already deduce that if $v$ is a leaf such that $c \notin \s(v)$, then no ancestor of $v$ in the support tree can possess $c$. This is a property similar to character state changes in the Camin-Sokal parsimony method~\cite{inferringphylos}. We thus define the following subset of nodes that are forced to not have $c$:
\begin{equation*}
    F_c(G, \s) = \{v \in V(G) : \exists w \in L(\sup{G}(v)) \mbox{ such that $c \notin \sigma(w)$ } \}
\end{equation*}

If $G$ and $\s$ are clear from the context, we will write $F_c$ instead of $F_c(G, \s)$.
Recall that for a network $G$, $R_v(G)$ denotes the set of nodes reachable from $v$.  The following characterization of PTNs will let us recognize them easily.

\begin{lemma}\label{lem:g-fc}
Let $G$ be a tree-based network with $\S$-map $\s$. Then $(G, \s)$ is a perfect transfer network for $\S$ if and only if for every character $c \in \C$, $G-F_c$ contains a node $v$ such that $R_v(G - F_c)$ contains every leaf in $L(G - F_c)$.  
\end{lemma}

\begin{proof}
$(\Rightarrow)$ Let $l$ be a $\C-$labeling of $G$ that explains $\S$, and let $c \in \C$.  
By Lemma~\ref{lem:Fzero}, we must have $F_c \subseteq \ov{V}_c(l)$.  The resulting graph $G-F_c$ thus contains all the nodes in $V_c(l)$. 
Due to the single origin requirement of our definition, $G - F_c$ contains a node $v$ that is able to reach all the leaves in $V_c(l)$, which includes all the leaves of $G - F_c$.

$(\Leftarrow)$ 
To show the converse,  
we build a $\C-$labeling $l$ in the following way: for every $c \in \C$, let $v$ be a node such that its reachable set in $G-F_c$ contains every leaf in $L(G-F_c)$, and let us call $v$ the \emph{origin} of $c$.  Then for $u \in V(G)$, we put $c \in l(u)$ if and only if $u \in R_v(G-F_c)$.  That is, the origin of $c$ is $v$ and it is transmitted to every node that $v$ reaches in $G - F_c$. 
We claim that $l$ satisfies all the conditions of Theorem \ref{thm:connect}.

Let $u \in L(G)$ and let us first argue that $l(u) = \s(u)$.
Let $c \in \C$ and let $v$ be the chosen origin of $c$ in $l$.  If $c \notin \s(u)$, then $u \in F_c$ and $v$ does not reach $u$ in $G - F_c$ (because $u$ is simply not in $G - F_c$), and by our construction, $c \notin l(u)$.  If $c \in \s(u)$, then $u \notin F_c$ and $v$ reaches $u$ in $G - F_c$, in which case we put $c \in l(u)$.  It follows that $l(u) = \s(u)$.

We now argue on the connectedness requirements of the $G[V_c(l)]$ subgraphs.  Again, let $c \in \C$ and let $v$ be the chosen origin of $c$ in $l$.
Since $V_c(l)$ consists of nodes reachable from $v$, it is evident that $v$ will be the only node in $G[V_c(l)]$ whose in-degree is $0$.
It is also evident that $(G - F_c)[V_c(l)]$ is connected, since $V_c(l)$ only contains nodes reachable from $v$ in $G - F_c$.  This implies that $G[V_c(l)]$ is also connected, since we cannot disconnect the $V_c(l)$ subgraph by putting back the nodes of $F_c$.

It remains to argue the third condition of Theorem~\ref{thm:connect}.  First assume that the origin $v$ of $c$ is equal to $\rho(G)$.  Then every leaf must possess $c$, as otherwise if some leaf $u$ would satisfy $c \notin \s(c)$, then $u \in F_c$ and the root would also be in $F_c$, by its definition.  It follows that $F_c$ is empty.  Then our construction puts $V_c(l) = V(G)$ since $\rho(G)$ reaches every node in $G - F_c = G$.  In this case, the third condition of the theorem holds.  So we may assume that $v \neq \rho(G)$.  In this case, $v$ does not reach $\rho(G)$ in $G - F_c$ since no node reaches the root (except the root itself).  Thus we may assume that $\ov{V}_c(l) \neq \emptyset$ and contains the root. 
Suppose for a contradiction that $\sup{G}[\ov{V_c}(l)]$ is disconnected, i.e. there exist some node $u \in \ov{V}_c(l)$ that $\rho(G)$ cannot reach in $G[\ov{V}_c(l)]$. Then in $\sup{G}$, $u$ has an ancestor $w \in V_c(l)$, which implies that the origin, $v$, can reach $w$ in $G-F_c$. 
This implies that $w \notin F_c$, which in turn implies that no node  from $w$ to $u$ in $\sup{G}$ can be in $F_c$.  
This means that the path from $w$ to $u$ exists in $G-F_c$.  Thus in $G - F_c$, $v$ reaches $u$ through $w$, and so our construction of $l$ would have put $c\in l(u)$, a contradiction. Hence, $\sup{G}[\ov{V_c}(l)]$ is connected and contains $\rho(G)$.
\end{proof}

The proof of the previous lemma implies the following verification algorithm:

\vspace{3mm}

\begin{algorithm}[h]
\DontPrintSemicolon
\SetKwProg{Fn}{function}{}{}
\Fn{findLabeling($G, \s, \S$)}
{
    Set $l(v) = \s(v)$ for every leaf $v \in L(G)$.\;
    Set $l(v) = \emptyset$ for all $v \in V(G) \setminus L(G)$.\;       
    \While{$l$ not NULL and $\C$ is not empty}
    {
        Pick a character $c \in \C$.\;
        Compute $F_c(G, \s)$.\;
        Compute the graph $\widehat{G_c} = G - F_c(G, \s)$\;
        \eIf{$\widehat{G_c}$ is not connected}
        {Set $l = $ NULL.\;}
        {
        Let $X$ be the set of all nodes in $\widehat{G_c}$ with in-degree 0.\;
        \eIf{ there is $v \in X$ such that $L(\widehat{G_c}) \subseteq R_v(\widehat{G_c})$ \label{line:in-deg-zero}
        }{
            Set $l(u) = l(u) \cup \{c\}$ for all $u$ in $R_v(\widehat{G_c})$.\;
        }{
            Set $l =$ NULL.\;
        }
        }
        Remove $c$ from $\C$
    }
    return $l$.\;
}

  \caption{Check if a given tree-based network $G$ explains $\S$.}
  \label{alg:explain}
\end{algorithm}

\begin{theorem}\label{thm:display}
Algorithm \ref{alg:explain} correctly solves the {\sf PTN recognition} problem in time $O(|\C||V(G)|^2)$.
\end{theorem}

\begin{proof}
We first argue that the algorithm is correct.  By Lemma~\ref{lem:g-fc}, it suffices to check, for every $c \in \C$, that some node $v$ reaches every leaf in $G - F_c$.  Moreover, when this is the case, the proof of Lemma~\ref{lem:g-fc} shows that we can add $c$ to every node in $R_v(G - F_c)$ to obtain a labeling that explains $\S$.
If the algorithm finds such a $v$, this is exactly the labeling that it applies.  So we must argue that when such a $v$ exists, the algorithm will find it.
 
For that, we claim that the verification in line~\ref{line:in-deg-zero} of Algorithm \ref{alg:explain} is enough to find the node required by Lemma \ref{lem:g-fc}.  That is, we do not need to check every node $v$ of $\widehat{G_c} = G - F_c$, but only those in the set $X$, i.e. the nodes of in-degree $0$.
Suppose that there is $y \in V(\widehat{G_c})$ that reaches every leaf in $L(\widehat{G_c})$.  If $y \in X$, we are done.  Otherwise, because $\widehat{G_c}$ is acyclic, there must be $v \in X$ that reaches $y$ (it can be found by following by iteratively following in-neighbors starting from $y$ until such a node is reached).  It follows that $R_y(\widehat{G_c}) \subseteq R_v(\widehat{G_c})$ and that $v$ also reaches every leaf.
Thus, it suffices to check every source in $\widehat{G_c}$.

Let us now argue the complexity.
For a character $c \in \C$, computing $F_c$ can be done in a post-order traversal of $\sup{G}$ in time $O(|V(G)|)$. Checking whether the resulting network is connected or not can be done in linear time too. Finding the nodes of in-degree $0$ takes linear-time and, for each $O(|V(G)|)$ of these nodes, computing $R_v(\widehat{G_c})$ takes time $O(|V(G)|)$. Thus each character requires time $O(|V(G)|^2)$, and thus our algorithm solves the problem in $O(|\C||V(G)|^2)$.
\end{proof}

\subsection*{The tree-completion problem}

We now turn to the problem of predicting transfer locations on a given species tree.
As mentioned in \cite{Pons2018}, 
species trees depict vertical inheritance and thus serve as basis for our networks, and HGT events can be seen as additional evolutionary events that happen on the way.
This motivates the feasibility question: given a set of taxa $\S$ and a species tree $T$ on $\S$, can we add transfer arcs to $T$ so that the resulting network explains $\S$?  Moreover, can we do this so that the resulting network is time-consistent?
One possibility is that a \emph{universal tree-based network}, which contains all phylogenetic trees on a given set of $n$ leaves~\cite{bordewich2018universal}, could explain any set of characters.  However, a base tree needs to be specified in our model, and it is not clear that such a choice always exists in a universal network.

Nevertheless, it turns out that there is no feasibility problem.  Indeed, it is relatively easy to show that, for a set of taxa $\S$, any tree $T$ on $\S$ can be complemented with transfers to become a PTN.
One way to achieve this is to add a transfer edge between the parent edge of every pair of leaves, from which point it can be argued that the network is a PTN.
In fact, we can show a stronger statement: if $T$ is ``pre-labeled'' in any way that acquired characters are never lost, then we can add transfers to $T$ to explain $\S$ while preserving the given pre-labeling.

Formally, for a tree $T$, a $\C$-labeling $\lambda$ of $T$ is called a \emph{no-loss $\C$-labeling} if, for any edge $(u, v) \in E(T)$ and any $c \in \C$, it holds that $c \in l(u)$ implies $c \in l(v)$.
Recall that if $G$ is a tree-based network, then the base tree $T$ of $G$ is obtained by removing transfer arcs and suppressing subdivision nodes.  Hence, $V(T) \subseteq V(G)$, and we use $V(G) \cap V(T)$ to explicitly refer to the nodes of $G$ that are also in the base tree.
We have the following problem.

\vspace{3mm}

\noindent 
The {\sf Pre-labeled Tree Completion} problem\\
\textbf{Input.} A set of taxa $\S$ on characters $\C$, a tree $T$ with $\S$-map $\s$, and a no-loss $\C-$labeling $\lambda$ of $T$ such that $\lambda(v) = \s(v)$ for each $v \in L(T)$.\\
\textbf{Output.}  A time-consistent tree-based network $G$ such that $T$ is the base tree of $G$, and a $\C$-labeling $l$ of $G$ that explains $\S$ and such that $l(v) = \lambda(v)$ for every $v \in V(G) \cap V(T)$.

\vspace{3mm}

We will now prove that this is always possible, even with the time-consistency constraint.

One interest of allowing a pre-labeling is that one can start with any hypothesis on where the characters appeared on a tree, and transfers can explain that hypothesis.  
Perhaps the most natural pre-labeling is the Fitch-like labeling where, for each character $c$, we add $c$ in every maximal subtree whose leaves have the character.
To be more specific, if $v$ is a leaf, $\lambda(v) = \s(v)$ and otherwise, let $u,w$ be the children of $v$, then $\lambda(v) = \lambda(u) \cap \lambda(w)$.
This corresponds to the reasonable hypothesis that characters always appear at maximal clades that have it, and we provide and algorithm that can explain this.  Again, this is one example of a possible pre-labeling, and our algorithm can explain any other that has the no-loss property, even if characters are again after the lowest common ancestor of the leaves that have this character.

We define a transfer operation between two nodes $u$ and $v$ on a tree-based network $G$.  We write $G\trans{u,v}$ to denote the tree-based network obtained by subdividing the respective incoming edges of $u$ and $v$ in $\sup{G}$, thereby creating new parent $u'$ for $u$ and $v'$ for $v$, and adding the transfer edge $(u',v')$.

It is important to point out that in a no-loss labeling, there can be multiple ancestral species that acquire a specific character that for the first time.  In our algorithm, this property will also be maintained in the constructed network, and we will use the following notion:

\begin{definition}
Let $G$ be a tree-based network with a no-loss $\C-$labeling $\lambda$. Let $(u, v)$ be an edge of $\sup{G}$.
We will say that v is a \textbf{first-appearance node for} $c$ under $\lambda$ if it holds that $u \in \ov{V_c}(l)$ and every descendant of $v$ in $\sup{G}$ belongs to $V_c(l)$.
\end{definition}

We may say \emph{first-appearance for $c$} if $\lambda$ is clear from the context.
Note that if $\lambda$ is a no-loss labeling,  a first-appearance node for $c$ cannot have a first-appearance node for $c$ as a descendant, and hence first-appearance nodes are pairwise incomparable.

We can now describe our algorithm.  The first step is to make $G$ a copy of the given tree $T$, and then we add transfer edges to $G$.
Note that in our problem, the given tree has no time map, and deciding where to put the transfers on $T$ in a time-consistent manner becomes surprisingly complicated.  For this reason, before adding any transfer, we start by constructing a time consistent map $\tau$ for $G$.  It is easy to do this in such a way that no two vertices have the same time.

From that point, we look at each character $c$ and their set of first-appearance nodes $a_1, \ldots, a_k$, ordered in decreasing order of age, and we greedily connect them using transfers.  The $\tau$ that we constructed dictates the order of connections, in the sense that each $a_i$ is assumed to transfer $c$ to the younger $a_{i+1}$ node.  This is achieved by finding a descendant of $a_i$ that could have co-existed between $a_{i+1}$ and its parent.  The $\tau$ map is also used to assign a time to the nodes created by transfers.

\begin{algorithm}[h]
\DontPrintSemicolon
\SetKwProg{Fn}{function}{}{}
\Fn{TransferAdditionGreedy($T,S, \lambda$)}
  {
  Let $G = (V, E_S \cup E_T)$ be the tree-based network such that $V(G) = V(T)$, $E_S= E(T)$ and $E_T=\emptyset$.\;
  Let $l$ be the $\C$-labeling in which $l(v) = \lambda(v)$ for all $v \in V(G)$.\;  
  Let $\tau$ be a time-consistent map for $G$ in which every leaf $v$ has $\tau(v) = 0$ and every internal node has a distinct time.\;
  
  \For{ $c \in \C$
  }{
  Compute $A_c$, the set of first-appearance nodes of character $c$ in $\sup{G}$.\;
  Let $X_c= (a_1, a_2,\dots, a_k)$ be the ordering of $A_c$ such that $\tau(a_i) \geq \tau(a_{i+1})$ for all $i \in [k - 1]$.\;
  \For{$i\in [k-1]$}{
    Let $a'_i$ be the parent of $a_i$ in $\sup{G}$\;
    Let $a_{i + 1}'$ be the parent of $a_{i+1}$ in $\sup{G}$.\;
    \If{$a'_i$ 
    has no descendant $w$ in $\sup{G}$ such that $(w, a_{i+1}') \in E_T$}
    {
        Look for $(w',w) \in E(\sup{G}(a'_i))$ such that $\tg{w'} > \tg{a_{i+1}}$ and $\tg{w} \leq \tg{a_{i+1}}$\;
       
        Apply the transformation  $G\trans{w,a_{i+1}}$.\;\label{line:transformation}
        Let $\hat{w}$ and $\hat{a}_{i+1}$ be the new parents of $w$ and $a_{i+1}$, respectively, in $\sup{G}$\;
        
        Set $l(\hat{w}) = (l(w) \cap l(w')) \cup \{c\}$ and $l(\hat{a}_{i+1})=(l(a_{i+1}') \cap l(a_{i+1}))  \cup \{c\}$.\;\label{line:label-cap}
        
        Set $\tau(\hat{w}) = \tau(\hat{a}_{i+1}) = \dfrac{\min(\tau(w'), \tau(a'_{i+1})) + \tau(a_{i+1})}{2}$ \label{line:transfer-time}\;
    }
    \Else 
    {
        Add $c$ to $l(w)$ and $l(a'_{i+1})$ if not already present, where $w$ is a descendant of $a'_i$ in $\sup{G}$ such that $(w, a'_{i+1}) \in E_T$\;
    }
  }
  }
  return($G, l$)
}
  \caption{Place an edge between all the first-appearance trees.}
  \label{alg:greedy}
\end{algorithm}

\vspace{4mm}

\begin{figure}[h]
    \centering
    \includegraphics[width=0.9\textwidth]{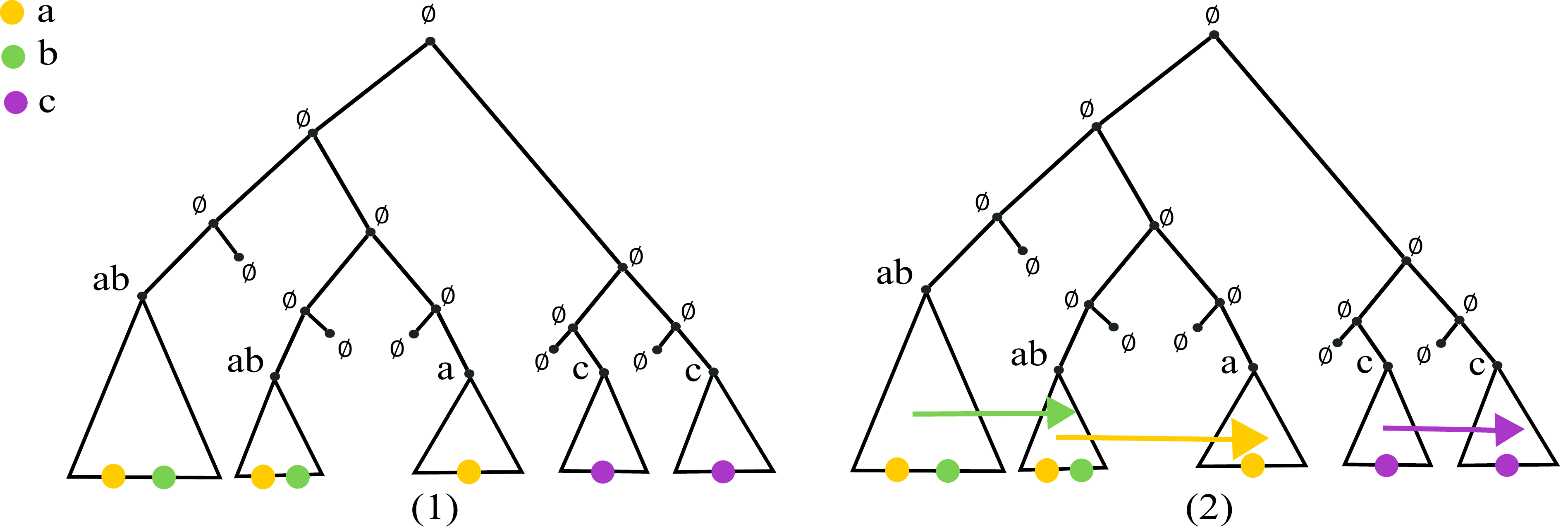}
    \caption{An example of a solution that will be given by our greedy algorithm (1) The given instance $T$ with $\C = \{a,b,c\}$. Triangles represent subtrees and colors represent the characters that can be found on them (2) One possible solution output by Algorithm~\ref{alg:greedy}. The green arrow joins the first two subtrees that contain $\{a,b\}$, the yellow arrows joins the second subtree with the subtree that contains only $\{a\}$ and the final arrow connects two subtrees that contain only $\{c\}$.}
    \label{fig:greedy}
\end{figure}

\begin{lemma}\label{lem:timec}
The network returned by Algorithm~\ref{alg:greedy} is time-consistent under the time map $\tau$.
\end{lemma}

\begin{proof}
We argue that at any moment during the execution of the algorithm, $\tau$ is a time-consistent map of $G$.
We prove this by induction on the number of iterations undertaken.
Notice that as a base case, the statement is initially true before entering the main \emph{for} loop.
Now assume that we have inserted $j - 1$ transfer edges and that $\tau$ is time-consistent for $G$ after these insertions.
Consider the $j$-th transfer edge inserted into $G$.
This transfer edge is $(\hat{w}, \hat{a}_{i+1})$, which are created between $w$ and its parent $w'$ in the support tree $T'$, and between $a_{i+1}$ and its parent $a'_{i+1}$ in $T'$, respectively.

We first claim that $(w', w)$ as used in the algorithm exists in the subtree $T(a_i)$, where $a_i$ is the node that precedes $a_{i+1}$ in $X_c$. We know that $\tau(a_i) \geq \tau(a_{i+1})$. Suppose that $\tau(a_i) = \tau(a_{i+1})$. Let $a_i'$ be the parent of $a_i$, then by induction hypotesis, since the network is time consistent we have that $\tau(a_i') > \tau(a_i) = \tau(a_{i+1})$ and so $w' = a_i'$ and $w = a_i$. In this case, we do have $\tau(w') > \tau(a_{i+1})$ and $\tau(w) \leq \tau(a_{i+1})$.  Now, suppose that $\tau(a_i) > \tau(a_{i+1})$. Note that we can always order the vertices of $T(a_i)$ with respect to $\tau$ in such a way that $\tau(a_i) \geq v$ for all $v \in V(T(a_i))$. By induction hypothesis, every time we pick a descendant $y$ of $a_i$, $\tau(y) < \tau(a_i)$, so $w$ can be found by iteratively following the descendants of $a_i$ and choosing the first one whose time is at most $\tau(a_{i+1})$ (which exists since all leaves have the same timing).

Note that by adding the transfer edge $(\hat{w}, \hat{a}_{i+1})$, the times $\tau$ of every edge have remained unchanged and still satisfy the time-consistency definition, with the exception of the edges linking $w', \hat{w}$, and $w$, those linking $a'_{i+1}, \hat{a}_{i+1}$, and $a_{i+1}$, as well as the new transfer edge. Since $\tg{\hat{w}} = \tg{\hat{a}_{i+1}}$ is made explicit on line~\ref{line:transfer-time}, 
to conclude the proof it suffices to show that $\tg{w'} > \tg{\hat{w}} > \tg{w}$ and that $\tg{a'_{i+1}} > \tg{\hat{a}_{i+1}} > \tg{a_{i+1}}$. By induction we know that $\tg{w'} > \tg{w}$ and that $\tg{a'_{i+1}} > \tg{a_{i+1}}$ (this holds before and after the transfer insertion because they were not changed). Using these inequalities we have that:
\begin{equation*}
\tg{\hat{w}} = \dfrac{\min(\tg{w'}, \tg{a'_{i+1}}) + \tg{a_{i+1}}}{2} < \dfrac{\tau(w') + \tau(w')}{2} = \tau(w')
\end{equation*}

since $\min(\tau(w'), \tau(a'_{i+1})) \leq \tau(w')$ and $\tau(a_{i+1}) < \tau(w')$ both hold.  Note that it is also true that

\begin{equation*}
\tg{\hat{w}} > \frac{\tau(a_{i+1}) + \tau(a_{i+1})}{2} = \tau(a_{i+1}) \geq \tau(w) 
\end{equation*}
since $\min(\tau(w'), \tau(a'_{i+1})) > \tau(a_{i+1})$ (because $\tau(w') > \tau(a_{i+1})$ and $\tau(a'_{i+1}) > \tau(a_{i+1})$ both hold). Using the same arguments, we have
\begin{equation*}
\tg{\hat{a}_{i+1}} = \dfrac{\min(\tg{w'}, \tg{a'_{i+1}}) + \tg{a_{i+1}}}{2} < \dfrac{\tau(a'_{i+1}) + \tau(a'_{i+1})}{2} = \tau(a'_{i+1})
\end{equation*}

and
\begin{equation*}
\tg{\hat{a}_{i+1}} > \frac{\tg{a_{i+1}}+ \tg{a_{i+1}}}{2} > \tg{a_{i+1}}    
\end{equation*}

\end{proof}

\begin{lemma}\label{lem:ok-labeling}
 Let $G$ and $l$ be the network and $\C$-labeling returned by Algorithm~\ref{alg:greedy}, respectively, on input $T, \S$, and $\lambda$.  Then $T$ is the base tree of $G$.  Furthermore,  $l$ explains $\S$ and satisfies $l(v) = \lambda(v)$ for every $v \in V(G) \cap V(T)$.
\end{lemma}

\begin{proof}
Let $G = (V, E_S \cup E_T)$ be the network returned by the algorithm and let $l$ be the returned labeling.
To see that $l(v) = \lambda(v)$ for every node $v$ in the base tree, it suffices to note that the algorithm never changes the labeling of a node initially present in $T$: it only assigns sets of characters to nodes created by transfer insertion operations.
It is also easy to see that $T$ is the base tree of $G$, since we start with a copy of $T$ and only attach transfer edges to it.

We will argue that the returned $\C-$labeling $l$ explains $\S$ using the conditions required by Definition~\ref{def:hgt-explain}. First, by requirement on the input $\lambda$, we know that for every $v \in L(G)$, $l(v) = \lambda(v) = \s(v)$.

Let us next show that for every character $c \in \C$ there exists a unique node $v \in V_c(l)$ that reaches every node in $G[V_c(l)]$.  It is not hard to see that after handling a particular $c \in \C$ in the algorithm, this property will be satisfied for $V_c(l)$.  However, it is not obvious that the subsequent iterations on other characters will not ``break'' this property for $c$.
We thus prove the following statement. Assume that the algorithm handles the characters of $\C$ in order $c_1, \ldots, c_m$ in the main $while$ loop.  Then we claim that in the network $G$ obtained after finishing the $i$-th iteration and handling $c_1, c_2, \ldots, c_i$, for every $j \leq i$, there exists a unique node $v \in V_{c_j}(l)$ that reaches every node in $G[V_{c_j}(l)]$.  This shows the desired property since it will hold for $j = m$, i.e. for every character.
As a base case, consider $i = 0$.  Then it is true that for $j \leq i$, the desired node $v$ exists (because there is no $c_j$ to satisfy).
Now, assume that $i > 0$ and that before entering the $i$-th iteration, the statement holds for every $c_j$ with $j \leq i - 1$.\\
After we are done handling $c_i$ on the $i-$th iteration we know that there exists a vertex $a_1 \in X_{c_i}$ such that $\tau(a_1) \geq \tau(a_h)$ for all $a_h \in X_{c_i}$. In particular, when equality holds for some $a_h$, it must be that $\tau(a_1) = \tau(a_2)$.  In this case, when the algorithm iterates on $a_1$, we get $(w', w) = (a'_1, a_1)$ and the corresponding transformation yields $G\trans{a_1,a_2}$. In this case, we claim that the vertex $\hat{a}_1$ created by the subdivision of the edge $(a_1',a_1) \in E_s$ will remain as the unique source for $V_{c_i}(l)$. To see this first note that $a_1$ has no incoming edge from $V_c(l)$ at the start of the iteration, because if that was not the case then $c \in l(a_1')$ and so $a_1$ would not have been a first-appearance node in the first place. This in turn implies that $\hat{a}_1$ has no incoming edge after the $i$-th iteration, since all other transfer heads are added above $a_2, \ldots, a_k$.  Additionally, $\hat{a}_1$ will become the first-appearance node for its corresponding subtree. After applying $G\trans{a_1,a_2}$, $\hat{a}_1$ now reaches the new parent $\hat{a}_2$ of $a_2$ and all its descendants in $\sup{G}$. Subsequently, we will now choose a descendant $w$ of $a_2$ from which we will add a new transfer to the parent node $\hat{a}_3$ of $a_3$. In this way, $\hat{a}_1$ will also reach $a_3$ and all of its descendants. We will continue adding transfer operations in this way so that finally $\hat{a}_1$ will be able to reach the last vertex of $X_{c_j}$, $a_k$ and all of its descendants. Note that any node of $V_{c_i}(l)$ is reachable by an element of $A_{c_i}$, thus after the $i-$th iteration $\hat{a}_1$ reaches every node in the $G[V_{c_i}(l)]$ subgraph.

On the other hand, when we have the strict inequality, i.e. when $\tau(a_1) > \tau(a_2)$, then the first chosen $w$ is distinct from $a_1$, and using the same arguments, we see that $a_1$ is the unique origin for $V_{c_i}(l)$.

We must also argue that the $i$-th iteration does not ``break'' a $c_j$ with $j < i$.  Consider such a $c_j$.  By induction, before the $i$-th iteration, $G[V_{c_j}(l)]$ had a unique origin $v$. 
Suppose that during the $i-$th iteration of the main loop we created some transformation $G\trans{w,a_h}$ after which some node, say $z$, that possesses $c_j$ cannot be reached by $v$ in the transformed graph.  Let $w', a'_h$ be the parents of $w$ and $a_h$, respectively, before the addition of the transfer.  Also let $\hat{w}$ and $\hat{a}_h$ be the nodes which were created by this transformation.

First assume that $z = \hat{w}$.  Then $c_j \notin l(w')$, as otherwise if $c_j \in l(w')$, by induction, $v$ would reach $w'$ and thus also reach $\hat{w} = z$.  But because $c_j \notin l(w')$, $c_j \notin l(w') \cap l(w)$ and so the algorithm would not have put $c_j$ in $\hat{w} = z$, a contradiction.  By the same argument, $z \neq \hat{a}_h$.  Thus, $z$ was present in $G$ before the insertion of the transfer.

Next, assume that $c_j \notin l(\hat{w})$ and $c_j \notin l(\hat{a}_h)$.  If $v$ cannot reach $z$ anymore, every path from $v$ to $z$ in $G[V_{c_j}(l)]$ must have been going through  $(w', w)$ or $(a'_h, a_h)$ before the transfer insertion.  But then, $c_j \in l(w') \cap l(w)$ and $c_j \in l(a'_h) \cap l(a_h)$, and the algorithm would have put $c_j \in l(\hat{w})$ and $c_j \in l(\hat{a}_h)$, a contradiction.

Then either $c_j \in l(\hat{w})$, $c_j \in l(\hat{a}_h)$ or $c_j \in l(\hat{w}) \cap l(\hat{a}_g)$.
Assume $c_j \in l(\hat{w})$. 
By line \ref{line:label-cap}, we know that this would only be possible if $c_j \in l(w')\cap l(w)$.
So any path from $v$ to $z$ in $V_{c_j}(l)$ that used the $(w', w)$ edge can now use the edges $(w', \hat{w}), (\hat{w}, w)$ to reach $z$.  If $c_j \in l(\hat{a}_h)$ as well, the same idea applies, and we get that any path from $v$ to $z$ is still usable, albeit with either $\hat{w}$ or $\hat{a}_h$ as an additional vertex.  So it must be that $c_j \notin l(\hat{a}_h)$, and that all paths used $(a'_h, a_h)$.  As before, this means that $c_j \in l(a'_h) \cap l(a_h)$ and that we should have $c_j \in l(\hat{a}_h)$, a contradiction.  This covers the case $c_j \in l(\hat{w})$.  
The case $c_j \in l(\hat{a}_h)$ can be handled in the same manner.

We then show that for each support edge $(u,v) \in E_S, c \in l(u)$ implies that $c \in l(v)$. 
We argue that this property holds before and after any transfer edge is inserted.
  Notice that initially, when $G$ is just a copy of $T$ and $l$ a copy of $\lambda$, $c \in l(u)$ implies $c \in l(v)$ because $\lambda$ is a no-loss labeling.
Now suppose inductively that the property holds before we insert some transfer $(\hat{w}, \hat{a}_{i+1})$ by line~\ref{line:transformation}.  It suffices to argue that the property holds on the support edges $(w', \hat{w}), (\hat{w}, w), (a'_{i+1}, \hat{a}_{i+1})$, and $(\hat{a}_{i+1}, a_{i+1})$, as defined in the algorithm, because no other support edge is modified.
Let $c' \in l(w')$ (we distinguish $c'$ from $c$, the latter being the $c$ the algorithm is currently iterating on).  Then by assumption that the property held before the transfer addition, we must have $c' \in l(w)$
and, because $c' \in l(w) \cap l(w')$, $c'$ will be added to $l(\hat{w})$, as desired.  The same argument holds for $c' \in l(a'_{i+1})$ and the fact that $c' \in l(\hat{a}_{i+1})$.  Now let $c' \in l(\hat{w})$.  We want to argue that $c' \in l(w)$.  If $c' \in l(w')$, then again by assumption we have $c' \in l(w)$ as well and our property holds.  So suppose that $c' \notin l(w')$.  
The algorithm puts $l(\hat{w}) = (l(w) \cap l(w')) \cup \{c\}$ where $c$ is the character of the current iteration, which means that only $c = c'$ is possible.
Notice that the algorithm chooses the edge $(w, w')$ in the subtree $\sup{G}(a'_i)$, where $a'_i$ has child $a_i$ that is a first-appearance node for $c$. By assumption, every descendant of $a_i$ in the support tree possesses $c$, so only $w' = a'_i$ is possible.  Thus $w = a_i$ and $c \in l(a_i) = l(w)$, as desired.  Finally, let $c' \in l(\hat{a}_{i+1})$. 
If $c' \in a'_{i+1}$, by assumption $c' \in l(a_{i+1})$ and we are done.  Otherwise, as the previous case we must have $c' = c$ and, since $a_{i+1}$ is a first-appearance for $c$, we have $c \in l(a_{i+1})$ as desired.
\end{proof}

\begin{theorem}\label{thm:greedyalg-is-ok}
Algorithm~\ref{alg:greedy} solves the {\sf Pre-labeled Tree Completion} problem correctly in time $O(|\C|^2|\S| + |\C||\S|^2)$.
\end{theorem}

\begin{proof}
By Lemma~\ref{lem:timec}, the network $G$ output by the algorithm has $T$ as its base tree and is time-consistent.
Then by Lemma~\ref{lem:ok-labeling}, the output labeling $l$ preserves the pre-labeling $\lambda$ and explains $\S$.  Thus, the output is correct.

As for the running time, the complexity is dominated by the main $while$ loop over $c \in \C$.
Note that first-appearance nodes partition the leaves of $G[V_c(l)]$, and so each $A_c$ has at most $|\S|$ elements and, for each $c$ we add at most $|\S|$ transfers.  It follows that the final network $G$ output by the algorithm has at most $O(|\C||\S| + |V(T)|)$ nodes.  Let us denote $n = |V(G)| \leq |\C||\S| + |V(T)|$.
For a given $c \in \C$, computing the first-appearance nodes can be done in time $O(n)$ and sorting them takes time $O(n \log n)$.
Then for each of the $O(n)$ nodes in $X_c$, we must find a descendant $w$ in time $O(n)$.  The other operations take constant time.
Thus, each iteration of the main loop takes time $O(n \log n + n^2) = O(n^2)$.  This is repeated $O(|\C|)$ times, and thus the complexity is $O(|\C|n^2) = O(|\C| \cdot (|\C||\S| + |V(T)|))$.
The claimed complexity follows from $|V(T)| \in O(|\S|)$ since $T$ is a tree on leafset $\S$.
\end{proof}

\subsection*{On minimizing the number of transfers in a completion}

We have shown that any tree whose leaves are mapped to $\S$ can explain $\S$.  The tree completion therefore becomes more interesting in the minimization variant:

\vspace{3mm}

\noindent 
The {\sf Minimum Perfect Transfer Completion} problem\\
\textbf{Input.} A set of taxa $\S$ on characters $\C$, a tree $T$ with $\S$-map $\s$.\\
\textbf{Output.}  A PTN $(G, \s)$ for $\S$ whose base tree is $T$ that contains a minimum number of transfer edges.

\vspace{3mm}

Note that the above problem does not impose a pre-labeling of $T$.
However, one such labeling $\lambda$ that is natural is the one where maximal subtree containing a character are assigned that character, which we call the Fitch labeling.

\begin{definition}
    Let $T$ be a tree with $\S$-map $\s$, where $\S$ are on characters $\C$.  The \emph{Fitch-labeling} of $T$ is the labeling $\lambda$ of $T$ such that, for each $c \in \C$, we put $c \in \lambda(v)$ if and only if all leaves descending from $v$ contain $c$.
\end{definition}

The Fitch-labeling can be combined with Algorithm~\ref{alg:greedy} to obtain basic bounds on the number of transfers required.

\begin{proposition}\label{prop:completion-bounds}
    Let $T$ be a tree with $\S$-map $\s$, where $\S$ is on characters $\C$.  
    Let $\lambda$ be the Fitch labeling for $T$, and for $c \in \C$, let $A_c$ be the set of first-appearance nodes of $c$ under $\lambda$.
    Then 
    \begin{itemize}
        \item 
        any PTN for $\S$ with base tree is $T$ requires at least $\max_{c \in \C} (|A_c| - 1)$ transfer edges;

        \item 
        there exists a PTN for $\S$ with base tree $T$ with at most $\sum_{c \in \C} (|A_c| - 1)$ transfer edges.  Moreover, Algorithm~\ref{alg:greedy} returns such a PTN when given pre-labeling $\lambda$.
    \end{itemize}
\end{proposition}

\begin{proof}
Let $(G, \s)$ be a PTN for $\S$ with base tree $T$.
Let $c \in \C$ be such that the number of first-appearance nodes in $A_c$ is maximum.  Notice that all nodes of $V(G) \cap V(T)$ that do not descend from a first-appearance node in $A_c$ cannot contain $c$ in any labeling of $G$ by Lemma~\ref{lem:Fzero} (since they have a descendant not possessing $c$, such a node is in $F_c$).  Therefore, any solution for $T$ must add at least $|A_c| - 1$ transfers to $T$ to be able to connect the first-appearance subtrees.  Thus $G$ requires at least $|A_c| - 1$ transfers. 

Now consider the output of Algorithm~\ref{alg:greedy} on pre-labeling $\lambda$.  For each $c \in \C$, the algorithm adds at most $|A_c| - 1$ transfer arcs when it considers $c$ in its $while$ loop (noting that the number of first-appearance nodes for $c$ never increases in the algorithm).  By Theorem~\ref{thm:greedyalg-is-ok}, the algorithm correctly returns a PTN, which has at most $\sum_{c \in \C}(|A_c| - 1)$ transfers.
\end{proof}

\begin{corollary}
Suppose that Algorithm~\ref{alg:greedy} is given $T, \S$, and the Fitch-labeling $\lambda$.  Then it is a $|\C|$-approximation, i.e. it outputs a PTN with at most $|\C|$ times more transfers than an optimal solution.
\end{corollary}

\begin{proof}
Using the notation of Proposition~\ref{prop:completion-bounds}, it states that an optimal PTN needs at least $|A_c| - 1$, and Algorithm~\ref{alg:greedy} adds at most $|\C|(|A_c| - 1)$ transfers. 
\end{proof}

Do note that 
    Algorithm~\ref{alg:greedy} does not always output an optimal solution to the {\sf Minimum Perfect Transfer Completion} problem.
To see that Algorithm~\ref{alg:greedy} can be suboptimal, consider Figure~\ref{fig:greedy-subopt} and the explanation in the caption.

\begin{figure}[h!]
    \centering
    \includegraphics[width=0.9\textwidth]{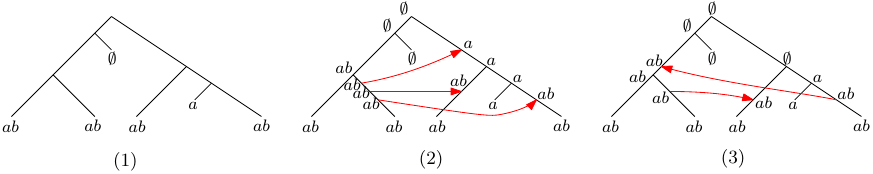}
    \caption{The greedy Algorithm~\ref{alg:greedy} can add more transfers than the minimum.  (1) A network $G$ whose leaves are on two characters $a$ and $b$. (2) One possible solution output by Algorithm~\ref{alg:greedy}.  The left clade is chosen to originate $a$ and a transfer is greedily added above all the $a$'s on the right side.  Then, the left clade is also chosen to originate $b$, but for that two more separate transfers are added.  (3) A solution that is somewhat less natural, but with one less transfer.}
    \label{fig:greedy-subopt}
\end{figure}

The cases in which Algorithm~\ref{alg:greedy} is suboptimal appear to have a common cause.  The algorithm tends to add transfers as high as possible in the tree, whereas the optmal solution would add more transfers lower in the tree, but with the advantage of being reusable by more characters.  For instance in Figure~\ref{fig:greedy-subopt}, $a$ is greedily transferred by itself and two more separate  transfers are required for $b$, whereas the optimal solution reuses the same transfers for both $a$ and $b$.
It appears difficult to determine when to add more transfers lower and when not to, which leads us conjecture that the {\sf Minimum Perfect Transfer Completion} is NP-hard.

We must reckon that a $|\C|$-approximation is far from interesting.
We conclude by suggesting a candidate approximation algorithm that combines both algorithms presented here.  Suppose that, given $T$ and Fitch-labeling $\lambda$, we obtain $G$ and $l$ from Algorithm~\ref{alg:greedy}.  A simple heuristic post-processing step can then be applied to detect unneeded transfers.
That is, for each transfer edge $(u, v)$ of $G$, consider the network $G - (u, v)$ obtained by removing $(u, v)$  and the resulting subdivision nodes.  Then, run Algorithm~\ref{alg:explain} to check if $G - (u, v)$ is a PTN for $\S$.  If so, we know that $(u, v)$ can safely be removed and we repeat.  We try every such transfer edge until all of them are necessary.
With this modified algorithm, we were unable to generate instances with more than twice as many transfers as the optimal solution.  This suggests that it might not be far from optimal.  
We conjecture that 
Algorithm~\ref{alg:greedy}, combined with the above post-processing step, achieves a constant factor approximation.

\subsection*{The PTN reconstruction problem}
We now study the variant in which only $\S$ and $\C$ are known, and no tree is given.  
Note that by Theorem~\ref{thm:greedyalg-is-ok}, there is no feasibility problem, since a solution always exists.  That is, we can take \emph{any} tree $T$ with any $\S$-map $\s$, run Algorithm~\ref{alg:greedy} using the Fitch-labeling, and obtain a PTN for $\S$.
The minimization variant has more appeal.

\vspace{3mm}

\noindent 
The {\sf Minimum Perfect Transfer Reconstruction} problem\\
\textbf{Input.} A set of taxa $\S$ on characters $\C$;\\
\textbf{Output.}  A PTN $(G, \s)$ for $\S$ with a minimum number of transfer edges.

\vspace{3mm}

This does not appear to be an easy algorithmic problem.  Proposition~\ref{prop:completion-bounds} suggests a (seemingly) simple approach: find a tree that minimizes the total number of first-appearance nodes, over all characters.  This does not guarantee that the resulting network will minimize transfers, and even finding a tree that minimizes the number of first-appearance nodes appears hard to find.  We conjecture that both problems are NP-hard, i.e. finding a PTN $(G, \s)$ for $\S$ with a minimum number of transfers, and finding a tree with $\S$-map $\s$ with a minimum number of first-appearance nodes.

In the following, we instead focus on providing bounds on the number of transfers required if $k := |\C|$ characters are present, in the worst case.  
It should be intuitive that the taxa set that will require the most transfers is when $\S = 2^{\C}$, i.e. $\S$ is the power set of $\C$.
When this is the case, we show that, up to a linear factor, an exponential number of transfers (with respect to $k$) is required and sufficient.

\begin{lemma}
    Any set of taxa $\S$ on a character set $\C$ of $k$ characters can  be explained by a tree-based network $G$ that has at most $2^k - (k-1)$ transfers.    
\end{lemma}

\begin{proof}
    Let $\S$ be a set of taxa on characters $\C = \{c_1, c_2, \ldots, c_k\}$.  We will assume that no two elements of $\S$ are identical (as two identical taxa can use the same transfers).  We will further assume that $\S = 2^{\C}$ contains every subset of $\C$.  Note that if we can explain $\S$ with at most $2^k - (k - 1)$ transfers, then we can explain any $\S' \subseteq \S$ with at most that many transfers.
    We will first show how to construct a base tree $T$ for $G$ from which we will then derive $G$. Let $T$ be a rooted complete binary tree with exactly $2^k$ leaves and root $\rho$. Let $l_b :V(T) \to \{0,1\}$ be a labeling such that for every inner node $u \in V(T)$ with $v,w$ as children we have that $l_b(v) = 0$ and $l_b(w) = 1$ (the label of the root is not important). For an example of this construction, see Figure \ref{fig:lemma16}.

    Let $\mathcal{L} = \{L_1, L_2, \dots , L_k\}$ be a partition of $V(T) \setminus \{\rho\} $ such that a vertex $u \in L_i$ with $i \in [k]$ if and only if the unique path from $u$ to the root contains exactly $i$ edges. We will call each $L_i$ a \emph{level} of $T$. 

    \begin{figure}
        \centering
        \includegraphics[width=0.85\textwidth]{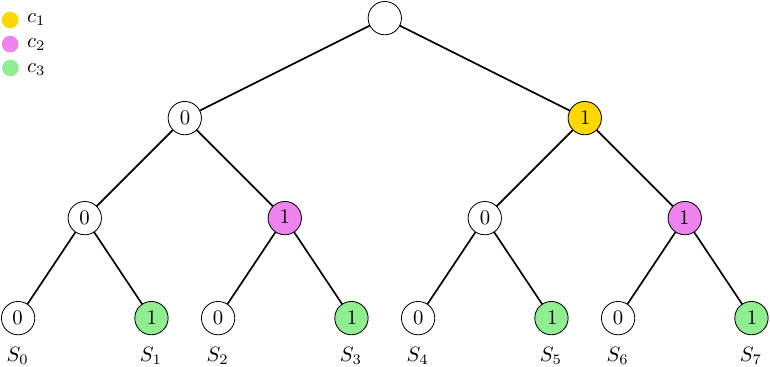}
        \caption{An illustration of $T$ with $k = 3$ and $|\S| = 8$.  Each color corresponds to a distinct character.  Each node labeled as $1$ is a first-appearance of the character associated with that color.  Notice that each level has its corresponding character.}
        \label{fig:lemma16}
    \end{figure}

    Consider the $\C$-labeling $l$ such that, 
    for each $i \in [k]$ and $u \in L_i$ with $l_b(u) = 1$, $u$ is a first-appearance of $c_i$.  Note that such a $\C$-labeling can be achieved by adding character $c_i$ to each node descending from every such $u$, and not adding $c_i$ to any other node. 
    For each $c_i \in \C$, we denote by $A_i$ its set of first-appearance nodes under $l$.
    Notice that $|A_i| = 2^{i-1}$, since $L_i$ has exactly $2^i$ nodes and, by definition, there exist $2^{i-1}$ nodes labeled as $1$.

    We argue that under $l$, the leaves of $T$ are in bijection with $\S$, and that $l$ can serve as the $\S$-map for $T$.  
    There are $2^k = |\S|$ leaves, so it suffices to show that for any distinct $u, v \in L(T)$, $l(u) \neq l(v)$.  Let $z = \lca_T(u, v)$ and let $z_1, z_2$ be the two children of $z$.  Note that $z_1, z_2$ are in the same level, say $L_i$ for some $i \in [k]$.  Moreover, $l_b(z_1) \neq l_b(z_2)$, implying that one of $u$ has $c_i$ and the other does not.  Then $l(u) \neq l(v)$.
    Finally, we will create a set of transfer edges $E_T$ for $T$ in the following way: we use Algorithm~\ref{alg:greedy} and give it $l$ as a pre-labeling.  
    Recall that for each character $c_i$, the algorithm adds at most $|A_i| - 1 = 2^{i-1} - 1$ transfer edges. 
    Since $2^{i-1} - 1$ transfers will be added to $T$ per level $L_i$ with $i \in \{1, 2, \ldots, k\}$, 
    this results in at most $$ \sum_{i=1}^{k} (2^{i-1} - 1) = 2^k - k - 1$$
    added transfers.
\end{proof}

\begin{lemma}
Let $\C$ be a set of $k$ characters.  Then there exists a set of taxa $\S$ on characters $\C$ such that
any PTN for $\S$ has at least $2^k/(3k) - 1$ transfer edges. 
\end{lemma}

\begin{proof}
Let $\S = 2^{\C}$ be the power set of $\C$.
Let $(G, \s)$ be a PTN for $\S$ and let $T$ be the base tree of $G$.  Let $\lambda$ be Fitch-labeling of $T$.  In what follows, all first-appearance nodes of $T$ are assumed to be with respect to $\lambda$.

Let $A(T)$ be the set of nodes of $T$ that are a first-appearance node of at least one character of $\C$.  We argue that $|A(T)| \geq 2^k/3$.
To achieve this, we describe a partial injective function $f : L(T) \rightarrow A(T)$ that maps some leaves of $T$ to a first-appearance node.  
By partial injective function, we mean that perhaps not all leaves are mapped, but no two mapped leaves map to the same node.   
A \emph{cherry} of $T$ is a pair of leaves that have the same parent.  A \emph{non-cherry} leaf is a leaf that is not in a cherry.
Note that since $L(T)$ is in bijection with $\S$ we will refer to the leaves of $T$ as taxa, with the understanding that leaves are subsets of $\C$.

First let $x, y$ be leaves that form a cherry in $T$.  Then since all taxa are distinct, $x$, $y$ disagree on at least one character $c$, and hence 
one of $x$ or $y$ is a first-appearance node of $c$, let us say $x$ without loss of generality.
Then we put $f(x) = x$ (and leave $f(y)$ undefined).
Now let $N'$ be the set of non-cherry leaves of $T$.  For each $x \in N'$ such that $x$ is a first-appearance node for some character, put $f(x) = x$.  

Next let $N = N' \setminus A(T)$, i.e. $N$ contains non-cherry leaves that are not a first-appearance node, for any character.
For $x \in N$, denote by $p_x$ the parent of $x$ in $T$.  Let us denote the set of nodes $\{p_x : x \in N\}$ as \emph{marked nodes}.  
For two distinct marked nodes $p_x, p_y$, we say that $p_y$ is a \emph{closest marked descendant} of $p_x$ if $p_y \prec p_x$ and there is no marked node 
on the path between $p_x$ and $p_y$, except $p_x$ and $p_y$ themselves.
Suppose that $x \in N$ is such that $p_x$ has at least one closest marked descendant, say $p_y$.  We argue that there is a first-appearance node $z$ satisfying $p_y \preceq z \prec p_x$. 
This situation is illustrated in Figure~\ref{fig:markednodes}.

\begin{figure}
    \centering
    \includegraphics[width=0.28\textwidth]{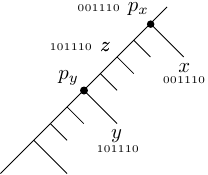}
    \caption{An illustration of $f(x) = z$ for a marked node $p_x$ that has a closest marked descendant $p_y$.
    Here, the characters of $x$ must be a subset of the characters of $y$, as otherwise $x$ would be a first appearance.  We can then argue that a character of $y \setminus x$ has its first appearance between $p_x$ and $p_y$.}
    \label{fig:markednodes}
\end{figure}

To see this, notice that $p_x$ is marked because $x$ is not a first-appearance.  This means that for every $c$ in $x$, all descending leaves of $p_x$ contain $c$.  This includes $y$, and so $x \subseteq y$.  Since $x$ and $y$ have distinct characters, there must be a character $c'$ such that $c'$ is in $y$ but $c'$ is not in $x$.  
Since $y$ is not a first-appearance of $c'$, all descendants of $p_y$ have $c'$.  Hence, there is a first-appearance node $z$ of $c'$ that is an ancestor of $p_y$, but a strict descendant of $p_x$ (because not all descendants of $p_x$ have $c'$).  We put $f(x) = z$.
Note that there may be multiple choices for $z$, in which case we choose arbitrarily.  The important property is that whichever choice is made, $z$ is a strict descendant of $p_x$, and not a strict descendant of any marked $p_y \prec p_x$.
This completes the description of $f$.

Suppose that $p_x, p_{x'}$ are distinct marked nodes that both have a closest marked descendant.  Notice that $f(x)$ and $f(x')$ must be distinct.  Indeed, if $p_x$ and $p_{x'}$ are incomparable, this is obvious since $f(x)$ and $f(x')$ are strict descendants of $p_x$ and $p_{x'}$, respectively.  If instead $p_x \prec p_{x'}$, then $f(x)$ is a strict descendant of $p_x$, whereas $f(x')$ is either a closest marked descendant of $p_{x'}$, or an ancestor of one of those nodes, implying that $f(x')$ cannot be a strict descendant of $p_x$.  
One can thus see that $f$ is injective, since it either maps leaves to themselves, or children of marked nodes to distinct internal nodes.

Now consider the set of leaves $x \in N$ such that $p_x$ has no marked descendants.  Note that all these marked $p_x$ nodes are pairwise incomparable.  Moreover, any such $p_x$ has a leaf child $x$, and another child with at least two leaves (if the other child was a single leaf, $p_x$ would form a cherry and $x$ would not be in $N'$).  Since any tree with at least two leaves has a cherry, $p_x$ has a descending cherry, and in fact each $p_x$ is the ancestor of a distinct cherry.  It follows that the number of $p_x$ with no marked descendant is at most the number of cherries.

We can finally relate the number of first-appearances to the number of taxa as follows.  
Let $C$ be the number of cherries of $T$; let $N_A$ be the number of non-cherry leaves that are first-appearances; let $N_M$ be the number of leaves in $N$ whose parent has a marked descendant; and let $N_X$ be the number of leaves in $N$ whose parent has no marked descendant.
Then 
\[
|L(T)| = 2C + N_A + N_M + N_X
\]
and, owing to our partial injection $f$, 
\[
|A(T)| \geq C + N_A + N_M
\]
We have argued just above that $N_X \leq C$, and so we can infer the chain of inequalities
\begin{align*}
C + N_A + N_M &= \frac{2C}{3} + N_A + N_M + \frac{C}{3} \\ 
			  &\geq \frac{2C}{3} + N_A + N_M + \frac{N_X}{3} \\
			  &\geq \frac{1}{3} (2C + N_A + N_M + N_X) \\
			  &\geq \frac{1}{3} |L(T)|
\end{align*}
and we get $|A(T)| \geq |L(T)|/3 = 2^{k}/3$.  

Each node of $A(T)$ is a first-appearance of some $c \in \C$.  
Therefore, by the generalized pigeonhole principle, some character $c$ of $\C$ must have at least $|A(T)|/|\C| \geq 2^{k}/(3k)$ first-appearance nodes in $T$.
By Proposition~\ref{prop:completion-bounds}, at least that many transfers, minus $1$, need to be added to $T$ to explain $\S$.
\end{proof}

\subsection*{Open problems}
We conclude this section with some open problems regarding all problems mentioned in this paper.

\begin{itemize}
    \item 
    Our algorithm for the recognition problem takes time $O(|\C||V(G)|^2)$.  
    Can this be improved to $O(|\C||V(G)|)$, or even $O(|\C| + |V(G)|)$?
    
    \item 
    Is the minimum tree completion problem NP-hard?

    \item 
    Does the greedy Algorithm~\ref{alg:greedy} achieve a constant factor approximation, with the post-processing mentioned in the tree completion section?

    \item 
    Is the minimum perfect transfer reconstruction NP-hard?

    \item 
    Suppose that the taxa set is $\S = 2^{\C}$ for characters $\C$.  Then what exactly is the minimum number of transfers of a PTN for $\S$?

    \item 
    How can our model be extended to characters with multiple states?  And would the underlying recognition problem be easy?
    
\end{itemize}

\section*{Conclusion}
In this contribution, we have introduced \emph{ perfect transfer networks} (PTNs)
as a model that aims to combine the structural properties of tree-based networks with the classical notion of perfect characters. In the absence of HGT events, PTNs coincide with the notion of perfect phylogenies. To better understand these, we studied their properties, stated the main differences between them and recombination networks as well as perfect phylogeny networks which is to our knowledge the closest model related to ours. Additionally, we explored several algorithmic challenges that result from this model with potential applications for HGT inference methods using character-based information that does not rely on sequence similarities. To the best of our knowledge, this is the initial theoretical endeavor employing the ``once acquired never lost'' principle through tree-based networks for inferring HGT events. While we acknowledge the simplicity of this model, it represents an initial stride towards incorporating additional conditions. Our intention is to develop more intricate models that better align with the complexity of biological systems. Although widely used throughout the literature, perfect characters impose strong restrictions for our model, since each character is allowed to change of state at most once. A potential extension of our model would be to include different models for character state changes as in \emph{Dollo parsimony} \cite{dolloparsimony}. 
This model allows losing an acquired character but not gaining it twice. This assumption has been shown to be more suitable for complex characters such as restriction sites and introns \cite{inferringphylos}. Another possible extension of our model would be to include expression levels of the different characters instead of discrete changes which is a problem similar to that of multi-state perfect phylogenies \cite{Gusfield2010}. 

Several other questions seem to be appealing for future work. Most importantly, since finding experimental datasets that use gene expression seems like a challenging task, the generation of \emph{in-silico} datasets to test our algorithms seems to be a pertinent solution. Nevertheless, to our knowledge there is no simulation framework that combines evolutionary histories with gene expression data. Therefore, a future direction for this project could also be the design of a simulation environment that is able to generate this type of data.

\section*{Acknowledgements}
The authors would like to thank the reviewers for their helpful comments and for pointing out paper~\cite{nakhleh}.

\section*{Funding}
Alitzel L\'opez S\'anchez acknowledges financial support from the programme de bourses d'excellence en recherche from the University of Sherbrooke.

Manuel Lafond acknowledges financial support from the Natural Sciences and Engineering Research Council (NSERC) and the Fonds de Recherche du Québec Nature et technologies (FRQNT)

\clearpage
\bibliographystyle{plain}
\bibliography{references}

\begin{thebibliography}{10}

\bibitem{Alexander2007}
Patrick~A Alexander, Yanan He, Yihong Chen, John Orban, and Philip~N Bryan.
\newblock The design and characterization of two proteins with 88\% sequence
  identity but different structure and function.
\newblock {\em Proceedings of the National Academy of Sciences},
  104(29):11963--11968, 2007.

\bibitem{AltafUlAmin2019}
Md. Altaf-Ul-Amin, Shigehiko Kanaya, and Zeti-Azura Mohamed-Hussein.
\newblock Investigating metabolic pathways and networks.
\newblock In {\em Encyclopedia of Bioinformatics and Computational Biology},
  pages 489--503. Elsevier, 2019.

\bibitem{anselmetti2021gene}
Yoann Anselmetti, Nadia El-Mabrouk, Manuel Lafond, and A{\"\i}da Ouangraoua.
\newblock Gene tree and species tree reconciliation with endosymbiotic gene
  transfer.
\newblock {\em Bioinformatics}, 37(Supplement\_1):i120--i132, 2021.

\bibitem{avni2020new}
Eliran Avni and Sagi Snir.
\newblock A new phylogenomic approach for quantifying horizontal gene transfer
  trends in prokaryotes.
\newblock {\em Scientific reports}, 10(1):1--14, 2020.

\bibitem{bafna2003haplotyping}
Vineet Bafna, Dan Gusfield, Giuseppe Lancia, and Shibu Yooseph.
\newblock Haplotyping as perfect phylogeny: A direct approach.
\newblock {\em Journal of Computational Biology}, 10(3-4):323--340, 2003.

\bibitem{bansal2012efficient}
Mukul~S Bansal, Eric~J Alm, and Manolis Kellis.
\newblock Efficient algorithms for the reconciliation problem with gene
  duplication, horizontal transfer and loss.
\newblock {\em Bioinformatics}, 28(12):i283--i291, 2012.

\bibitem{bodlaender1992two}
Hans~L Bodlaender, Mike~R Fellows, and Tandy~J Warnow.
\newblock Two strikes against perfect phylogeny.
\newblock In {\em International Colloquium on Automata, Languages, and
  Programming}, pages 273--283. Springer, 1992.

\bibitem{bordewich2018universal}
Magnus Bordewich and Charles Semple.
\newblock A universal tree-based network with the minimum number of
  reticulations.
\newblock {\em Discrete Applied Mathematics}, 250:357--362, 2018.

\bibitem{boto2010horizontal}
Luis Boto.
\newblock Horizontal gene transfer in evolution: facts and challenges.
\newblock {\em Proceedings of the Royal Society B: Biological Sciences},
  277(1683):819--827, 2010.

\bibitem{Bourque2018}
Guillaume Bourque, Kathleen~H. Burns, Mary Gehring, Vera Gorbunova, Andrei
  Seluanov, Molly Hammell, Michaël Imbeault, Zsuzsanna Izsv{\'{a}}k, Henry~L.
  Levin, Todd~S. Macfarlan, Dixie~L. Mager, and C{\'{e}}dric Feschotte.
\newblock Ten things you should know about transposable elements.
\newblock {\em Genome Biology}, 19(1), November 2018.

\bibitem{Camin1965}
Joseph~H. Camin and Robert~R. Sokal.
\newblock A method for deducing branching sequences in phylogeny.
\newblock {\em Evolution}, 19(3):311, September 1965.

\bibitem{Cardona2015reconstruction}
Gabriel Cardona, Joan~Carles Pons, and Francesc Rosselló.
\newblock A reconstruction problem for a class of phylogenetic networks with
  lateral gene transfers.
\newblock {\em Algorithms for Molecular Biology}, 10(1), December 2015.

\bibitem{de1995phenotypic}
Gjalt De~Jong.
\newblock Phenotypic plasticity as a product of selection in a variable
  environment.
\newblock {\em The American Naturalist}, 145(4):493--512, 1995.

\bibitem{delabre2020evolution}
Matt{\'e}o Delabre, Nadia El-Mabrouk, Katharina~T Huber, Manuel Lafond, Vincent
  Moulton, Emmanuel Noutahi, and Miguel~Sautie Castellanos.
\newblock Evolution through segmental duplications and losses: a
  super-reconciliation approach.
\newblock {\em Algorithms for Molecular Biology}, 15(1):1--15, 2020.

\bibitem{della2017character}
Gianluca Della~Vedova, Murray Patterson, Raffaella Rizzi, and Mauricio Soto.
\newblock Character-based phylogeny construction and its application to tumor
  evolution.
\newblock In {\em Conference on Computability in Europe}, pages 3--13.
  Springer, 2017.

\bibitem{doyon2010efficient}
Jean-Philippe Doyon, Celine Scornavacca, K~Yu Gorbunov, Gergely~J
  Sz{\"o}ll{\H{o}}si, Vincent Ranwez, and Vincent Berry.
\newblock An efficient algorithm for gene/species trees parsimonious
  reconciliation with losses, duplications and transfers.
\newblock In {\em RECOMB international workshop on comparative genomics}, pages
  93--108. Springer, 2010.

\bibitem{dolloparsimony}
James~S. Farris.
\newblock {Phylogenetic Analysis Under Dollo's Law}.
\newblock {\em Systematic Biology}, 26(1):77--88, 03 1977.

\bibitem{inferringphylos}
Joseph Felsenstein.
\newblock {\em Inferring phylogenies}.
\newblock Sunderland, Mass. : Sinauer Associates, 2004.

\bibitem{fernandez2001perfect}
David Fern{\'a}ndez-Baca.
\newblock The perfect phylogeny problem.
\newblock In {\em Steiner Trees in Industry}, pages 203--234. Springer, 2001.

\bibitem{FRANCIS201893}
Andrew Francis, Charles Semple, and Mike Steel.
\newblock New characterisations of tree-based networks and proximity measures.
\newblock {\em Advances in Applied Mathematics}, 93:93--107, 2018.

\bibitem{francis2015phylogenetic}
Andrew~R Francis and Mike Steel.
\newblock Which phylogenetic networks are merely trees with additional arcs?
\newblock {\em Systematic Biology}, 64(5):768--777, 2015.

\bibitem{geiss2018reconstructing}
Manuela Gei{\ss}, John Anders, Peter~F Stadler, Nicolas Wieseke, and Marc
  Hellmuth.
\newblock Reconstructing gene trees from fitch’s xenology relation.
\newblock {\em Journal of Mathematical Biology}, 77(5):1459--1491, 2018.

\bibitem{Goyal2022}
Akshit Goyal.
\newblock Horizontal gene transfer drives the evolution of dependencies in
  bacteria.
\newblock {\em {iScience}}, 25(5):104312, May 2022.

\bibitem{Gusfield2010}
Dan Gusfield.
\newblock The multi-state perfect phylogeny problem with missing and removable
  data: Solutions via integer-programming and chordal graph theory.
\newblock {\em Journal of Computational Biology}, 17(3):383--399, March 2010.

\bibitem{gusfield2014recombinatorics}
Dan Gusfield.
\newblock {\em ReCombinatorics: the algorithmics of ancestral recombination
  graphs and explicit phylogenetic networks}.
\newblock MIT press, 2014.

\bibitem{gusfield2004optimal}
Dan Gusfield, Satish Eddhu, and Charles Langley.
\newblock Optimal, efficient reconstruction of phylogenetic networks with
  constrained recombination.
\newblock {\em Journal of Bioinformatics and Computational Biology},
  2(01):173--213, 2004.

\bibitem{hellmuth2019reconciling}
Marc Hellmuth, Katharina~T Huber, and Vincent Moulton.
\newblock Reconciling event-labeled gene trees with mul-trees and species
  networks.
\newblock {\em Journal of Mathematical Biology}, 79(5):1885--1925, 2019.

\bibitem{hellmuth2020generalized}
Marc Hellmuth, Carsten~R Seemann, and Peter~F Stadler.
\newblock Generalized fitch graphs {II}: Sets of binary relations that are
  explained by edge-labeled trees.
\newblock {\em Discrete Applied Mathematics}, 283:495--511, 2020.

\bibitem{hotopp2011horizontal}
Julie C~Dunning Hotopp.
\newblock Horizontal gene transfer between bacteria and animals.
\newblock {\em Trends in genetics}, 27(4):157--163, 2011.

\bibitem{iersel2019third}
Leo~Van Iersel, Mark Jones, and Steven Kelk.
\newblock A third strike against perfect phylogeny.
\newblock {\em Systematic Biology}, 68(5):814--827, 2019.

\bibitem{irwin2022systematic}
Nicholas~AT Irwin, Alexandros~A Pittis, Thomas~A Richards, and Patrick~J
  Keeling.
\newblock Systematic evaluation of horizontal gene transfer between eukaryotes
  and viruses.
\newblock {\em Nature microbiology}, 7(2):327--336, 2022.

\bibitem{jacox2017resolution}
Edwin Jacox, Mathias Weller, Eric Tannier, and Celine Scornavacca.
\newblock Resolution and reconciliation of non-binary gene trees with
  transfers, duplications and losses.
\newblock {\em Bioinformatics}, 33(7):980--987, 2017.

\bibitem{jones2022consistency}
Mark Jones, Manuel Lafond, and Celine Scornavacca.
\newblock Consistency of orthology and paralogy constraints in the presence of
  gene transfers.
\newblock {\em Peer Community in Mathematical and Computational Biology}, 2012.

\bibitem{keeling2008horizontal}
Patrick~J Keeling and Jeffrey~D Palmer.
\newblock Horizontal gene transfer in eukaryotic evolution.
\newblock {\em Nature Reviews Genetics}, 9(8):605--618, 2008.

\bibitem{koonin2001horizontal}
Eugene~V Koonin, Kira~S Makarova, and L~Aravind.
\newblock Horizontal gene transfer in prokaryotes: quantification and
  classification.
\newblock {\em Annual Reviews in Microbiology}, 55(1):709--742, 2001.

\bibitem{kordi2015complexity}
Misagh Kordi and Mukul~S Bansal.
\newblock On the complexity of duplication-transfer-loss reconciliation with
  non-binary gene trees.
\newblock {\em IEEE/ACM Transactions on Computational Biology and
  Bioinformatics}, 14(3):587--599, 2015.

\bibitem{lafond2020reconstruction}
Manuel Lafond and Marc Hellmuth.
\newblock Reconstruction of time-consistent species trees.
\newblock {\em Algorithms for Molecular Biology}, 15(1):1--27, 2020.

\bibitem{malikic2019phiscs}
Salem Malikic, Farid~Rashidi Mehrabadi, Simone Ciccolella, Md~Khaledur Rahman,
  Camir Ricketts, Ehsan Haghshenas, Daniel Seidman, Faraz Hach, Iman
  Hajirasouliha, and S~Cenk Sahinalp.
\newblock Phiscs: a combinatorial approach for subperfect tumor phylogeny
  reconstruction via integrative use of single-cell and bulk sequencing data.
\newblock {\em Genome research}, 29(11):1860--1877, 2019.

\bibitem{murakami2021phylogenetic}
Yukihiro Murakami.
\newblock {\em On Phylogenetic Encodings and Orchard Networks}.
\newblock PhD thesis, TU Delft, 2021.

\bibitem{nakhlehtesis}
Luay Nakhleh.
\newblock {\em Phylogenetic networks}.
\newblock PhD thesis, The University of Texas at Austin, 2004.

\bibitem{nakhleh}
Luay Nakhleh, Don Ringe, and Tandy Warnow.
\newblock Perfect phylogenetic networks: A new methodology for reconstructing
  the evolutionary history of natural languages.
\newblock {\em Language}, 81(2):382--420, 2005.

\bibitem{Pons2018}
Joan~Carles Pons, Charles Semple, and Mike Steel.
\newblock Tree-based networks: characterisations, metrics, and support trees.
\newblock {\em Journal of Mathematical Biology}, 78(4):899--918, oct 2018.

\bibitem{pons2019tree}
Joan~Carles Pons, Charles Semple, and Mike Steel.
\newblock Tree-based networks: characterisations, metrics, and support trees.
\newblock {\em Journal of Mathematical Biology}, 78(4):899--918, 2019.

\bibitem{pontes2013configurable}
Beatriz Pontes, Ra{\'u}l Gir{\'a}ldez, and Jes{\'u}s~S Aguilar-Ruiz.
\newblock Configurable pattern-based evolutionary biclustering of gene
  expression data.
\newblock {\em Algorithms for Molecular Biology}, 8(1):1--22, 2013.

\bibitem{pradhan2018non}
Dikshant Pradhan and Mohammed El-Kebir.
\newblock On the non-uniqueness of solutions to the perfect phylogeny mixture
  problem.
\newblock In {\em RECOMB International Conference on Comparative Genomics},
  pages 277--293. Springer, 2018.

\bibitem{ravenhall2015inferring}
Matt Ravenhall, Nives {\v{S}}kunca, Florent Lassalle, and Christophe Dessimoz.
\newblock Inferring horizontal gene transfer.
\newblock {\em PLoS Computational Biology}, 11(5):e1004095, 2015.

\bibitem{rawat2008novel}
Arun Rawat, Georg~J Seifert, and Youping Deng.
\newblock Novel implementation of conditional co-regulation by graph theory to
  derive co-expressed genes from microarray data.
\newblock In {\em BMC Bioinformatics}, volume~9, pages 1--9. Springer, 2008.

\bibitem{ringe2002indo}
Don Ringe, Tandy Warnow, and Ann Taylor.
\newblock Indo-european and computational cladistics.
\newblock {\em Transactions of the Philological Society}, 100(1):59--129, 2002.

\bibitem{sanderson1996homoplasy}
Michael~J Sanderson and Larry Hufford.
\newblock {\em Homoplasy: the recurrence of similarity in evolution}.
\newblock Elsevier, 1996.

\bibitem{sashittal2021parsimonious}
Palash Sashittal, Simone Zaccaria, and Mohammed El-Kebir.
\newblock Parsimonious clone tree reconciliation in cancer.
\newblock In {\em Leibniz International Proceedings in Informatics, LIPIcs},
  volume 201, page~9. Schloss Dagstuhl--Leibniz-Zentrum f{\"u}r Informatik,
  2021.

\bibitem{schaller2021indirect}
David Schaller, Manuel Lafond, Peter~F Stadler, Nicolas Wieseke, and Marc
  Hellmuth.
\newblock Indirect identification of horizontal gene transfer.
\newblock {\em Journal of Mathematical Biology}, 83(1):1--73, 2021.

\bibitem{semple2002tree}
Charles Semple and Mike Steel.
\newblock Tree reconstruction from multi-state characters.
\newblock {\em Advances in Applied Mathematics}, 28(2):169--184, 2002.

\bibitem{thomas2005mechanisms}
Christopher~M Thomas and Kaare~M Nielsen.
\newblock Mechanisms of, and barriers to, horizontal gene transfer between
  bacteria.
\newblock {\em Nature Reviews Microbiology}, 3(9):711--721, 2005.

\bibitem{tofigh2010simultaneous}
Ali Tofigh, Michael Hallett, and Jens Lagergren.
\newblock Simultaneous identification of duplications and lateral gene
  transfers.
\newblock {\em IEEE/ACM Transactions on Computational Biology and
  Bioinformatics}, 8(2):517--535, 2010.

\bibitem{trejo2012cancer}
Catalina Trejo-Becerril, Enrique P{\'e}rez-C{\'a}rdenas, Luc{\'\i}a
  Taja-Chayeb, Philippe Anker, Roberto Herrera-Goepfert, Luis~A
  Medina-Vel{\'a}zquez, Alfredo Hidalgo-Miranda, Delia P{\'e}rez-Montiel, Alma
  Ch{\'a}vez-Blanco, Judith Cruz-Vel{\'a}zquez, et~al.
\newblock Cancer progression mediated by horizontal gene transfer in an in vivo
  model.
\newblock {\em PloS One}, 7(12):e52754, 2012.

\bibitem{van2010quantifying}
Leo van Iersel, Charles Semple, and Mike Steel.
\newblock Quantifying the extent of lateral gene transfer required to avert a
  genome of eden.
\newblock {\em Bulletin of Mathematical Biology}, 72:1783--1798, 2010.

\bibitem{wang2001perfect}
Lusheng Wang, Kaizhong Zhang, and Louxin Zhang.
\newblock Perfect phylogenetic networks with recombination.
\newblock {\em Journal of Computational Biology}, 8(1):69--78, 2001.

\bibitem{Wells2020}
Jonathan~N. Wells and C{\'{e}}dric Feschotte.
\newblock A field guide to eukaryotic transposable elements.
\newblock {\em Annual Review of Genetics}, 54(1):539--561, November 2020.

\bibitem{Zachar2020}
Istv{\'{a}}n Zachar and Gergely Boza.
\newblock Endosymbiosis before eukaryotes: mitochondrial establishment in
  protoeukaryotes.
\newblock {\em Cellular and Molecular Life Sciences}, 77(18):3503--3523,
  February 2020.

\end{thebibliography}
\clearpage

\end{document}